\documentclass[a4paper,11pt]{article}
\usepackage{fullpage}

\usepackage{amsfonts,amssymb,amsmath}
\usepackage{graphicx}
\usepackage{hyperref}
\usepackage{graphics}
\usepackage[usenames,dvipsnames]{color}
\usepackage{algorithm}
\usepackage{algpseudocode}
\usepackage{bbm}

\newtheorem{definition}{Definition}
\newtheorem{theorem}{Theorem}
\newtheorem{lemma}{Lemma}

\newtheorem{claim}{\rm \emph{Claim}}
\newtheorem{proposition}{Proposition}

\newcommand{\qed}{\hfill $\Box$ \bigbreak}
\newenvironment{proof}{\noindent{\bf Proof.~}}{\qed}

\newcommand{\m}[1]{{\cal{#1}}} 
\newcommand{\opt}{\mbox{\sc opt}}
\newcommand{\nca}{\mbox{\rm nca}}
\newcommand{\apex}{\mbox{\rm apex}}
\newcommand{\lightsize}{\mbox{\rm lsize}}
\newcommand{\ID}{\mbox{\sc id}}
\newcommand{\ST}{\mbox{\sc st}}
\newcommand{\MDST}{\mbox{near-\sc mdst}}
\newcommand{\NP}{\mbox{\rm NP}}
\newcommand{\coNP}{\mbox{\rm co-NP}}
\newcommand{\ch}{\mbox{\rm child}}
\newcommand{\dist}{{\rm dist}}

\newenvironment{smallitemize} {
 \begin{list}{$-$} {\setlength{\parsep}{0pt}
\setlength{\itemsep}{0pt}} } { \end{list} }

\newcommand{\AlgoT}{\mbox{\bf{tree\&{}lead}}}
\newcommand{\AlgoA}{\mbox{\bf{switch}}}
\newcommand{\AlgoP}{\mbox{\bf{permute}}}

\newcommand{\Root}{\mbox{\sf r}}
\newcommand{\p}{\mbox{\sf p}}
\newcommand{\level}{\mbox{\sf d}}
\newcommand{\F}{\mbox{\sf f}}
\newcommand{\size}{\mbox{\sf s}}
\newcommand{\switch}{\mbox{\sf switch}}
\newcommand{\up}{\mbox{\sf up}}

\newcommand{\TRoot}{{\tt Root}}
\newcommand{\TNode}{{\tt Node}}
\newcommand{\ErT}{{\tt ErrorForest}}
\newcommand{\TUp}{{\tt Up}}
\newcommand{\TDown}{{\tt Down}}
\newcommand{\UpWave}{{\tt UpWave}}
\newcommand{\ErL}{{\tt ErrorSwitch}}

\newcommand{\Tsize}{{\rm Size}}
\newcommand{\Best}{{\rm Best}}

\newcommand{\R}{\mathbb{R}}

\begin{document}
%%%%%%%%%%%%%%%%%%%%%%%%%%%%%%%%%%%%

\title{Polynomial-Time  Space-Optimal Silent Self-Stabilizing \\ Minimum-Degree Spanning Tree Construction}

\date{}

\author{
L\'elia Blin\thanks{Additional supports from the ANR project IRIS. }\\
{\small Laboratory LIP6-UPMC}\\
{\small University of Evry-Val d'Essonne}\\
{\small France}
\and
Pierre Fraigniaud\thanks{Additional supports from the ANR project DISPLEXITY, and from the INRIA project GANG.}\\
{\small Laboratory LIAFA}\\
{\small CNRS and Univ. Paris Diderot}\\
{\small France}
}
 
\maketitle

\begin{abstract}
Motivated by applications to sensor networks, as well as to many other areas,  this paper studies the construction of minimum-degree spanning trees. We consider the classical node-register state model, with a weakly fair scheduler, and we present a space-optimal \emph{silent} self-stabilizing construction of minimum-degree spanning trees in this model. Computing a spanning tree with minimum degree is NP-hard. Therefore, we actually focus on constructing a spanning tree whose degree is within one from the optimal. Our algorithm uses registers on $O(\log n)$ bits, converges in a polynomial number of rounds, and performs polynomial-time computation at each node. Specifically, the algorithm constructs and stabilizes on a special class of spanning trees, with degree at most $\opt+1$. Indeed, we prove that, unless $\NP=\coNP$, there are no proof-labeling schemes involving polynomial-time computation at each node for the whole family of  spanning trees with degree at most $\opt+1$. Up to our knowledge, this is the first example of the design of a compact silent self-stabilizing algorithm constructing,  and stabilizing on a subset of optimal solutions to a natural problem for which there are no time-efficient proof-labeling schemes. On our way to design our algorithm, we establish a set of independent results that may have interest on their own. In particular, we describe a new space-optimal silent self-stabilizing spanning tree construction, stabilizing on \emph{any} spanning tree, in $O(n)$ rounds, and using just \emph{one} additional bit compared to the size of the labels used to certify trees.  We also design a silent loop-free  self-stabilizing algorithm for transforming a tree into another tree. Last but not least, we provide a silent  self-stabilizing algorithm for computing and certifying the labels of a NCA-labeling scheme. 
\end{abstract}

\thispagestyle{empty}
\setcounter{page}{0}
\newpage

%%%%%%%%%%%%%%%%%%%%%%%%%%%%%%%%%%%%%%%%%%%%%
\section{Introduction} 
%%%%%%%%%%%%%%%%%%%%%%%%%%%%%%%%%%%%%%%%%%%%%

%-------------------------------------------------------------------------
\subsection{Context and objective}
%-------------------------------------------------------------------------

Self-stabilization~\cite{Dol00} deals with the design and analysis of distributed algorithms in which processes are subject to \emph{transient} failures. The main objective of self-stabilization is to evaluate the capacity for an asynchronous distributed system to recover from a transient fault, that is, to measure the ability of the system to return to a legal  state starting from an \emph{arbitrary}  state, and to remain in legal  states whenever starting from a legal  state. The legality of a  state is a notion that depends on the task to be solved, like, e.g.,  for leader election, the presence of a unique leader. 

One desirable property for a self-stabilizing algorithm is to be \emph{silent}~\cite{DGS99}, that is, to keep the individual state of each process unchanged once a legal (global) state has been reached. Silentness is a desirable property as it guarantees that self-stabilization does not burden the system with extra traffic between processes whenever the system is in a legal  state. Designing silent algorithms is difficult because one must insure that each process is able to decide \emph{locally} of the legality of a (global) state of the system, based solely on its own individual state, and on the individual states of its neighbors. This difficulty becomes prominent when one takes into account an important complexity measure for self-stabilizing algorithms: \emph{space complexity}, i.e., the amount of memory used at each process to store its variables~\cite{BGJ99,DGS99}. Keeping the memory space limited at each process reduces the potential corruption of the memory, and enables to maintain several redundant copies of variables (e.g., for fault-tolerance) without hurting the efficiency of the system. Moreover, in the classical \emph{node-register state model} (i.e., the model used in this paper), keeping the space complexity small insures that reading  variables in registers of neighboring processes does not consume too much bandwidth. 

In this paper, we are interested in the design of silent self-stabilizing algorithms with small space complexity in the node-register state model (with a classical weakly fair scheduler). More specifically, we focus our attention on the design of self-stabilizing algorithms for the construction of \emph{spanning trees} in networks. In this context, each process is a node of a graph $G$, and the nodes communicate along the edges of $G$. For instance, in the node-register state model, every node has read/write access to its own variables, and read-only access to the variables of its neighbors in~$G$. The objective is to compute a spanning tree $T$ of $G$. Typically, the tree $T$ is rooted at some node~$r$, and it is distributedly encoded at each node $v$ by the identify of  $v$'s parent $p(v)$ in $T$. (The root $r$ has $p(r)=\bot$).

There is a huge literature on the self-stabilizing construction of various kinds of trees, including spanning trees (ST)~\cite{Cournier2009,KosowskiK2005},  breadth-first search (BFS) trees~\cite{AfekB1998,AfekKY1990,BurmanK2007,CournierRV2011,DolevIM1993,HuangC1992,Johnen1997}, depth-first search (DFS) trees~\cite{CollinD1994,CournierDV2009,CournierRV2011,HuangC1993}, minimum-weight spanning trees (MST)~\cite{BlinDPR10,BlinPRT2009,GuptaS2003,HighamL2001,KKM11}, shortest-path spanning trees~\cite{GuptaBS2000,Huang2005}, minimum-diameter spanning trees~\cite{ButelleLB1995}, minimum-degree spanning trees~\cite{BlinGR2011}, etc.  Some of these constructions are even silent, with optimal  space-complexity. This is for instance the case of several BFS constructions under different kinds of schedulers~\cite{AfekKY1990,BurmanK2007,CournierRV2011,HuangC1992}, and of the ST constructions in~\cite{Cournier2009,KosowskiK2005}. All these latter constructions insure silentness via the (potentially implicit) use of a mechanism known as \emph{proof-labeling scheme}~\cite{KKP10}. 

A proof-labeling scheme for a graph property $\m{P}$ assigns a \emph{label} to each node so that, given its own label, and the labels of its neighbors, each node can decide whether  $\m{P}$  holds of not. More precisely, if  $\m{P}$  holds, then all nodes must decide ``yes'', otherwise at least one node must decide ``no''. For instance, $\m{P}$ may be ``$T$ is a spanning tree of $G$'' for some subgraph $T$ of $G$ partially known at each node (e.g., using the aforementioned parent pointer $p(v)$ at each node $v$). There are known compact and time-efficient proof-labeling schemes for many different types of  trees, including ST~\cite{KKP10}, BFS~\cite{AfekKY1990}, and MST~\cite{KK07,KKP10}. Typically, in the context of self-stabilization, a node detecting some inconsistencies between its own label and the labels of its neighbors decides ``no'', and launches a recovery procedure~\cite{AKY97,APSV91}. This procedure must insure that the system returns to a legal  state, and, also, must recompute the appropriate labels so that to distributedly certify the legality of the newly computed  state. We address the silent and compact self-stabilizing construction of trees belonging to a family of trees for which it is  \emph{unlikely} that there exists a time-efficient proof-labeling scheme. 

More specifically, we  focus on the construction of \emph{minimum-degree} spanning trees. That is, we aim at designing an algorithm which, for any given (connected) graph $G$, constructs a spanning tree $T$ of $G$ whose degree\footnote{I.e., the maximum,  taken over all nodes $v$ of $T$, of the degree of $v$ in $T$.} is minimum among all spanning trees of $G$. Our interest for this problem is  motivated by resolving issues arising in the design of MAC protocols for sensor networks under the 802.15.4 specification, which we are currently investigating with our partners STElectronics and Thales, in the framework of the project IRIS~\cite{iris}. It is also worth pointing out that the minimum-degree spanning tree problem arises in many other contexts, including electrical circuits~\cite{NH80}, communication networks~\cite{F01}, as well as in many other areas \cite{G82,K02}. 

Since {\sc Hamiltonian-path} is NP-hard, we actually slightly relax  our task, by focussing on the construction of a spanning tree whose degree is within $+1$ from the minimum degree $\opt$ of any spanning tree in the given graph. (Designing an algorithm for constructing a spanning tree with degree $\opt$, involving polynomial-time computation at each node, and polynomially many rounds, is hopeless, unless $\mbox{P}=\NP$). 

Even if constructing a spanning tree with degree at most $\opt+1$ can be sequentially achieved in polynomial time~\cite{FR94}, designing even just a terminating distributed algorithm for this task is challenging. To see why, let $\MDST(G,T)$ be the predicate ``$T$ is a spanning tree of $G$ with degree at most ${\opt+1}$''. As we shall show in this paper, unless $\NP=\coNP$, there are \emph{no} proof-labeling schemes involving polynomial-time computation at each node for this predicate. This negative result led us to address the intriguing question of whether there exist a compact (i.e., logarithmic space) and time-efficient (i.e., polynomial  computation at each node, and polynomial number of rounds) silent self-stabilizing algorithm for the construction of spanning trees with degree within $+1$ from the optimal. Perhaps surprisingly, we answer positively to this latter question \emph{despite the lack of an efficient proof-labeling scheme for that class of trees}. 

%-------------------------------------------------------------------------
\subsection{Our results}
%-------------------------------------------------------------------------

We design (and prove correctness of) a silent self-stabilizing algorithm which constructs and maintains a spanning tree of the actual network, whose degree is within $+1$ from the minimum degree of any spanning tree of this network. The space complexity of the algorithm is $O(\log n)$ bits at each node in $n$-node networks, which is optimal as a direct consequence of~\cite{DGS99}. Starting from an arbitrary  state, the algorithm converges to a legal spanning tree (i.e., a spanning tree of degree at most $\opt +1$) in a polynomial number of rounds. Moreover, each step of the algorithm involves polynomial-time computation at every activated nodes. 

In fact, our algorithm constructs a special kind of trees, named \emph{FR-trees} after F\"urer and Raghava\-chari~\cite{FR94}. The algorithm, which is a distributed version of the algorithm in~\cite{FR94}, accepts only  those trees (see Fig.~\ref{fig:summary}). FR-trees are of  degree at most $\opt+1$, but the trees of  degree at most $\opt+1$ which are not FR-trees will be rejected by our algorithm. That is, even starting from an initial spanning tree  $T$ of degree $\leq \opt+1$, but different from a FR-tree, our algorithm  transforms $T$ into a FR-tree $T'\neq T$, and stabilizes on $T'$. This is because, as we show later, verifying whether a given tree is of degree $\leq \opt+1$ cannot be done in polynomial time, unless $\NP=\coNP$. Instead, there is a proof-labeling scheme for FR-trees using labels on $O(\log n)$ bits. 

\begin{figure}[tb]
\centerline{\includegraphics[width=10cm]{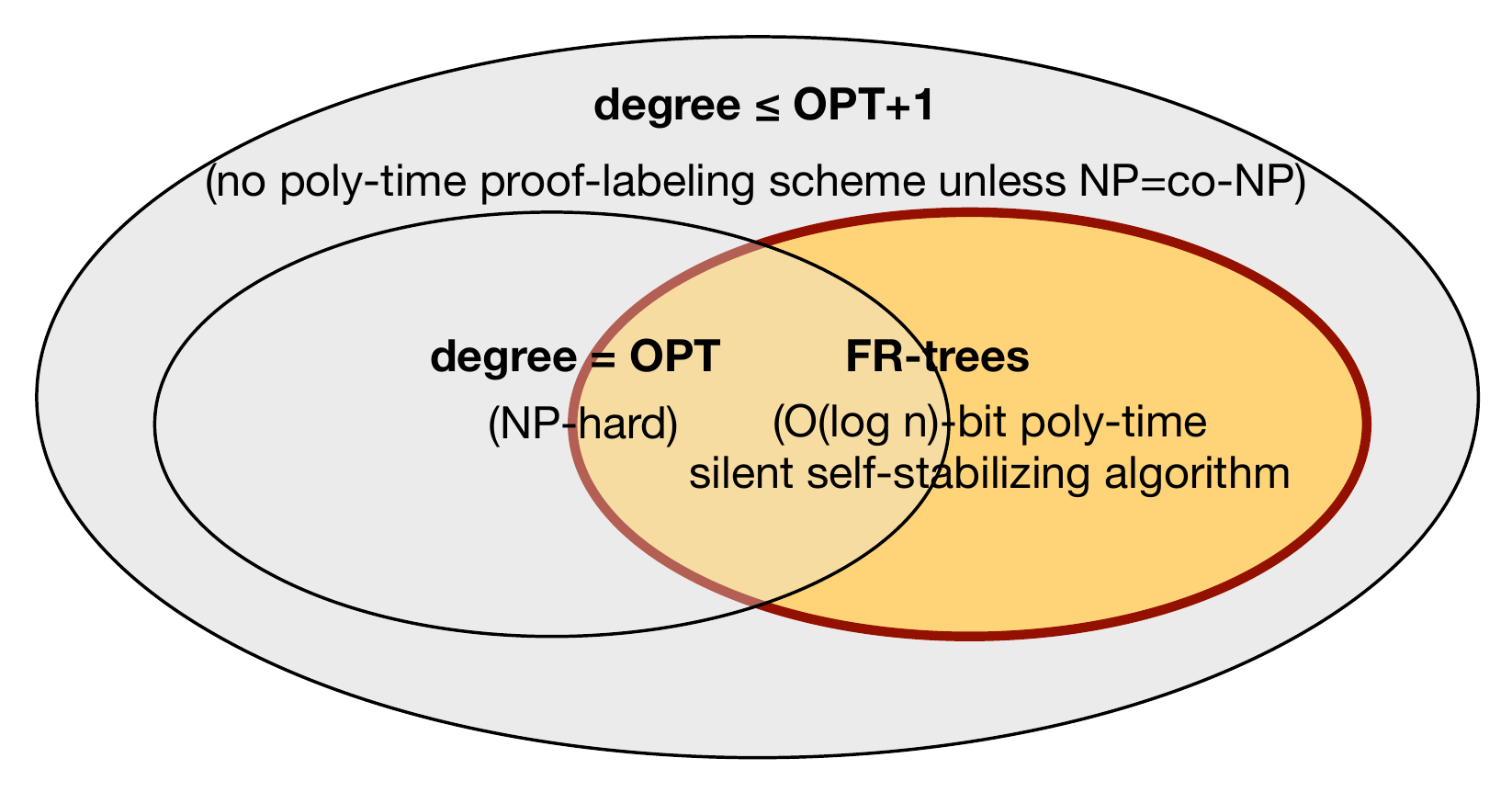}}
\caption{\sl Three classes of spanning trees. Our self-stabilizing algorithm constructs and stabilizes on FR-trees. The other two classes are hard for silent self-stabilization. }
\label{fig:summary}
\end{figure}

Constructing FR-trees is achieved using several techniques, which may have interest on their own. In particular, we present a novel silent self-stabilizing algorithm constructing, and stabilizing on \emph{any} spanning tree, converging in $O(n)$ rounds. In addition, the root of the tree becomes explicitly a leader. The algorithm is simple (only four rules), and elegant (it uses just \emph{one} additional bit compared to the size of the labels used to certify trees). We needed such an algorithm that accepts \emph{any} spanning tree because our minimum-degree spanning tree construction proceeds by successive improvements of the current tree, in order to decrease its degree, which may yield any tree-structures. Up to our knowledge, previous silent self-stabilizing spanning tree constructions either stabilize only on specific forms of trees (e.g., BFS), or are designed under the \emph{semi-uniform} model (i.e., assuming the presence of a leader). None of these constructions were thus well suited for contexts similar to ours,  in which the trees may evolve arbitrarily. 

We also present a silent \emph{loop-free} self-stabilizing algorithm for  transforming any tree into another tree (previous loop-free algorithms are not silent). Our silent loop-free algorithm is based on a redundant way to label the nodes of a tree so that one can update a spanning tree, say from $T_1$ to $T_2$, while preserving the ability to certify the spanning tree property at any point in time during the transformation. (Instead, classical proof-labeling schemes for trees requires to recompute the labels after the transformation is completed, and these labels may be incorrect during the transformation). 

Finally, our minimum-degree spanning tree algorithm  makes use of the informative-labeling scheme for nearest common ancestor (NCA) from~\cite{AGKR04}. Up to our knowledge, this is the first time that this very compact NCA-labeling scheme is used in the context of self-stabilization.  Moreover, using this scheme in a silent algorithm requires to design a proof-labeling scheme for it. It is probably the first occurrence of a proof-labeling scheme for an informative-labeling scheme!

%-------------------------------------------------------------------------
\subsection{Other related work}
%-------------------------------------------------------------------------

In addition to the aforementioned references, there is a series of contributions that are closely related to our work. In particular, several papers address the leader election task~\cite{AroraG1990,Awerbuch1993,BGJ99,DattaLV2011,DolevH1997}, which is inherently related to spanning tree construction. Regarding the sequential construction of minimum-degree spanning tree, the best known result is \cite{FR94} which describes a polynomial-time algorithm for constructing spanning as well as Steiner trees with degree at most $\opt+1$. Several generalizations of the minimum-degree spanning tree problem have been addressed, including the degree-bounded minimum-weight spanning tree problem~\cite{SL07}, and minimum-degree  spanning tree problem in digraphs~\cite{KKRR04}. 

%%%%%%%%%%%%%%%%%%%%%%%%%%%%%%%%%%%%%%%%%%%%%
\section{The computational model} 
%%%%%%%%%%%%%%%%%%%%%%%%%%%%%%%%%%%%%%%%%%%%%

In this paper, we are dealing with the \emph{state model} for self-stabilization, where each process is the node of an asynchronous network modeled as a simple connected graph $G=(V,E)$. Every node $v\in V$ has a distinct identity, denoted by $\ID(v)\in\{1,\dots,n^c\}$ for some constant $c\geq 1$, and is a state machine with state-set $S$ (the same for all nodes).  The machine has read/write access to a single-writer multiple-reader register which stores the current state of $v$. (The identity $\ID(v)$ of every node $v$ is a constant, which does not necessarily appear in the register of node $v$, unless $v$ explicitly writes it in there). In one atomic \emph{step}, every node can read its own register, and the registers of its neighbors in $G$. The transition function is a function $\delta: S^* \to S$ which, given any finite collection of states, returns a new state. At node $v$ in state $s_0$, and given the states $s_1,\dots,s_k$ of the $k$ neighbors of $v$ in $G$, the new state $s'_0$ of $v$ after one step  is $s'_0=\delta(s_0,\{s_1,\dots,s_k\})$. The network is asynchronous in the sense that nodes take step of computation (i.e., change state) in arbitrary order, under the control of a \emph{weakly fair} scheduler. That is, at each step, the scheduler must choose at least one of the enabled node (those for which the algorithm wants to take a step) under the constraint that every node enabled persistently must eventually take a step. 

A collection of $n$ individual registers-state in an $n$-node graph form a (global)  \emph{ state} of the system. A \emph{problem},  or a  \emph{task}, is specified by a set of   states, called \emph{legal} states. For instance, in the case of spanning tree construction, one wants each node $v$ to maintain a variable $p(v)$ storing either the identity of its parent, or $\bot$. A  state is legal if and only if the 1-factor defined by the set of (directed) edges
\[
\{(v,p(v)),v\in V\}
\] 
form a spanning tree of $G$. A \emph{fault} is the corruption of the register of one or more nodes in the system. After a fault has occurred, the system may be in an illegal  state. It is the role of the self-stabilizing algorithm to detect the illegality of the current  state, and to make sure that the system returns to a legal  state. In other words, starting from any  state, the system must eventually converge to a legal one, and must remain in legal ones. A self-stabilizing algorithm is \emph{silent} if and only if it converges to a  state where the values of the registers used by the algorithm remain fixed. Note that, for variables storing information whose size may vary depending on the structure of the network, like, typically, the identity of a leader, the corruption of that variable cannot result is storing a value with arbitrary large size. For instance, in an $n$-node network, a variable $X_v$ storing a node identity at node $v$ may be corrupted but this corruption can only result in having $X_v$ arbitrary in $[1,n^c]$ since this is what we assumed to be the range of valid identities in an $n$-node network.  

In this paper, we present a collection of algorithms whose combination will result in an algorithm constructing a minimum-degree spanning tree. For some of these algorithms, we even present an implementation of them. In this case, we adopt the classical way of describing implementations of self-stabilizing algorithms. Each node executes the same instruction set which consists in one or more \emph{rules} of the form: 
\[
\mbox{name-of-rule} : guard \longrightarrow command
\]
where $guard$ is a boolean predicate over the variables in the registers of the node and its neighbors, and $command$ is a  statement assigning new values to the variables of the node. An \emph{enabled}, or \emph{activatable}  node is a node for which at least one guard is true. Note that all implementations described in this paper satisfies that at most one guard is true at any node, at any point in time of the execution. 

Given a  state $\gamma$ of the system, let $A\in V$ be the set of nodes for which at least one guard of some algorithm {\tt Alg} is true. That is, $A$ is the set of activatable nodes in $\gamma$. A \emph{round} of an execution $\cal E$ of {\tt Alg} starting from $\gamma$ is the shortest  prefix of $\cal E$ in which each node in $A$ executes at least one step. If {\tt Alg} constructs and stabilizes on  states in some family $F$ of  states, then the round-complexity of {\tt Alg} is the maximum, taken over all initial states $\gamma$, and over all executions $\cal E$ of {\tt Alg} starting from~$\gamma$ and ending in a state $\gamma'\in F$, of the number of rounds in $\cal E$. The latter is the integer $k$ such that $\cal E$ can be decomposed in a sequence $\gamma_0=\gamma, \gamma_1, \dots,\gamma_{k}=\gamma'$ such that, for every $i=0,\dots,k-1$, the round of $\cal E$ starting from $\gamma_i$ ends in $\gamma_{i+1}$.  

%%%%%%%%%%%%%%%%%%%%%%%%%%%%%%%%%%%%%%%%%%%%%
\section{Minimum-degree spanning trees} 
%%%%%%%%%%%%%%%%%%%%%%%%%%%%%%%%%%%%%%%%%%%%%

%-------------------------------------------------------------------------
\subsection{The minimum-degree spanning tree  (MDST) problem}
%-------------------------------------------------------------------------

Given a tree $T$, the \emph{degree} $\deg(T)$ of $T$ is the maximum, taken over all nodes $v$ of $T$, of the degree of $v$ in $T$. Given a graph $G=(V,E)$, a spanning tree $T$ of $G$ is of minimum degree if there are no spanning trees $T'$ of $G$ with $\deg(T')<\deg(T)$. We denote by $\opt(G)$ the degree of any minimum-degree spanning tree of $G$.  Since deciding whether a graph is Hamiltonian is NP-hard, we get that deciding, given $G$ and $k\geq 0$, whether $\opt(G)\leq k$  is NP-hard. However, thanks to the Algorithm by F\"urer and Raghavachari~\cite{FR94}, given any graph $G$, one can construct a spanning tree $T$ of $G$ with $\deg(T)\leq \opt(G)+1$, in polynomial time. 

This paper describes a self-stabilizing distributed algorithm which, whenever running in a network~$G$, returns a \emph{\MDST} of $G$. That is, the algorithm returns a spanning tree $T$ of $G$ with $\deg(T)\leq \opt(G)+1$. This spanning tree is encoded distributedly as follows: it is rooted at an arbitrary node $r$,  and every node $v$ stores the identity $p(v)$ of its parent in $T$ (the root $r$ stores $p(r)=\bot$). When the algorithm stabilizes, the 1-factor $\{(v,p(v)),v\in V\}$ must be a \emph{\MDST} of the current network. 

%-------------------------------------------------------------------------
\subsection{Proof-labeling scheme for minimum-degree spanning trees}
\label{subsec:pls4mdst}
%-------------------------------------------------------------------------

A silent algorithm needs to detect locally, by having each node inspecting only its register and the registers of its neighbors, whether the current state is legal or not. A typical mechanism for doing so is \emph{proof-labeling scheme}. To every node is assigned a \emph{label} so that, given its own label, and the labels of its neighbors, each node can decide whether a certain property holds or not. If it holds, then all nodes must decide ``yes'', otherwise at least one node must decide ``no''. For instance, in the case of spanning tree construction, for any given (connected) graph $G$, let us define 
\[
\ST(G)=\{T:\mbox{$T$ is a spanning tree of $G$}\},
\]
The following proof-labeling scheme for $\ST$, based on a sequence of increasing integers is folklore (see \cite{IL94} for other certificates for cycle-freeness). The label $L(v)$ of node $v$ is a pair $(\ID,d)$ where $\ID$ is the identity of the root $r$ of $T$, and $d$ is the distance of $v$ to that root in $T$. Each node $v$ checks that its given root identity $\ID$ is identical to the root identity given to all its neighbors in $G$, and checks that the distance given to its parent $p(v)$ is one less than the distance $d$ given to it (the root $r$ check that $d=0$). We call this scheme \emph{distance-based}. It is easy to see that if $T$ is not a spanning tree of $G$, that is, if $T$ is not spanning all nodes of $G$, or if $T$ is not a tree ($T$ may be a forest, or $T$ may contain a cycle), then some inconsistencies will be detected at some nodes, for every given collection $\{L(v),v\in V(G)\}$ of labels. The distance-based scheme uses labels on $O(\log n)$ bits, and the verification performed at each node runs in polynomial time -- it merely consists of at most $k+1=O(n)$ integer comparisons at each node of degree~$k$. 

The situation is radically different for $\MDST$, where
\[
\MDST(G)=\{T:\mbox{$T$ is a spanning tree of $G$, and $\deg(T)\leq \opt(G)+1$}\}.
\]
Indeed, the following results shows that it is unlikely that there is a proof-labeling scheme for $\MDST$ using  labels of logarithmic size. 

\begin{proposition}\label{prop:noPLS}
Unless $\NP=\coNP$, there are no proof-labeling schemes for $\MDST$ involving $O(poly(n))$ computation time at each node of $n$-node graphs. Thus, in particular, unless $\NP=\coNP$, there are no proof-labeling schemes for $\MDST$ using labels of size $O(\log n)$ bits.
\end{proposition}

\begin{proof}
Assume that  there exists a proof-labeling schemes for $\MDST$ involving polynomial-time computation at each node. Given a graph $G$, and a spanning tree $T$ with degree at most $\opt(G)+1$, the labels given by the scheme form a distributed certificate proving that $T\in \MDST(G)$. Since the verification scheme performs in $O(poly(n))$ time at each node, one can assume, w.l.o.g., that the labels are on $O(poly(n))$ bits since the verifier inspects $O(poly(n))$ bits  of the labels anyway. One can thus create a global certificate of polynomial size by concatenating all the $n$ labels. A sequential algorithm can then use this global certificate, and simulate the actions of each node. This sequential simulation amounts to polynomial time as well. Therefore, we get that deciding whether $T\in \MDST(G)$ is in NP. 

Let $\m{A}$ be a verification algorithm for $T\in \MDST(G)$. We now focus on the language {\sc no-Hamiltonian-path} consisting of all non-Hamiltonian graphs. Using the fact that deciding whether $T\in \MDST(G)$ is in NP, we show that {\sc no-Hamiltonian-path} is in NP for graphs of degree~4. Given a graph $G$ of degree~4 in $\mbox{\sc no-Hamiltonian-path}$, the certificate consists in an arbitrary spanning tree $T$ of $G$ with degree exactly~4, and the certificate $C$ used to prove that $T \in \MDST(G)$ using $\m{A}$. Since $G$ contains no Hamiltonian paths, we have $T$ of degree at most $\opt+1$, and therefore $\m{A}$ is accepting $T$ with certificate $C$. If $G\notin \mbox{\sc no-Hamiltonian-path}$, then $G$ contains a Hamiltonian path, and any spanning tree $T$ of degree~4 has degree larger than $\opt+1$. This leads $\m{A}$ to reject this instance for every given certificate $C$. Therefore, {\sc no-Hamiltonian-path} is in NP for graphs of degree~4. 

Now, {\sc Hamiltonian-path} is known to be NP-complete for grid graphs~\cite{IPS82}, hence for graphs with degree~3, and thus also for graphs with degree~4 (pick a node $u$ of degree~3 connected to $x,y$, and $z$, and replace $u$ by two nodes $u'$ and $u''$ connected by an edge, and both connected to $x,y$, and $z$). Therefore, {\sc no-Hamiltonian-path} is co-NP-complete for graphs with degree~4. Since we have just seen that $\mbox{\sc no-Hamiltonian-path}\in \NP$, we get that $\coNP\subseteq \NP$. Therefore, if there exists a proof-labeling schemes for $\MDST$ involving polynomial-time computation at each node, then $\coNP\subseteq \NP$, and thus $\NP=\coNP$.
\end{proof}

A direct consequence of the above result is that it is unlikely that there exists a silent, time-efficient self-stabilizing algorithm constructing and stabilizing on spanning trees with degrees at most $\opt+1$. Therefore, we shall now focus on constructing trees belonging to a subclass of spanning trees with degrees at most $\opt+1$ (see Fig.~\ref{fig:summary}).

%-------------------------------------------------------------------------
\subsection{A subclass of minimum-degree spanning trees}
%-------------------------------------------------------------------------

We define FR-trees, named after F\"urer and Raghavachari. 

\begin{definition}\label{def:FR-tree}
$T$ is a \emph{FR-tree} in a graph $G$ if $T$ is a degree-$k$ spanning tree of  $G$ whose every node can be marked ``good'' or ``bad'' such that the following three properties hold: 
%\begin{itemize}
(1)~every node with degree $k$ in $T$ is marked  bad,
(2)~every bad node is of degree at least $k-1$ in $T$, and 
(3)~there are no edges in $G$ between two good nodes in two different trees in the forest resulting from $T$ by removing the bad nodes (and their incident edges). 
%\end{itemize}
\end{definition}

F\"urer and Raghavachari~\cite{FR94} have proved that every (connected) graph has a spanning tree that is a FR-tree (e.g., an Hamiltonian path in an Hamiltonian graph is a FR-tree by marking all nodes bad). Theorem 2.2 in~\cite{FR94} states that the degree $k$ of any FR-tree satisfies $k\leq \opt+1$. However, not all spanning trees of degree $\opt$ or $\opt+1$ are FR-trees, as exemplified on Fig.~\ref{fig:FRtree}. The algorithm of F\"urer and Raghavachari, of which we shall provide a distributed self-stabilizing implementation in the next section, produces a FR-tree. In order to keep the algorithm silent, we use a proof-labeling scheme for FR-trees. 

\begin{figure}[tb]
\centerline{\includegraphics[width=8cm]{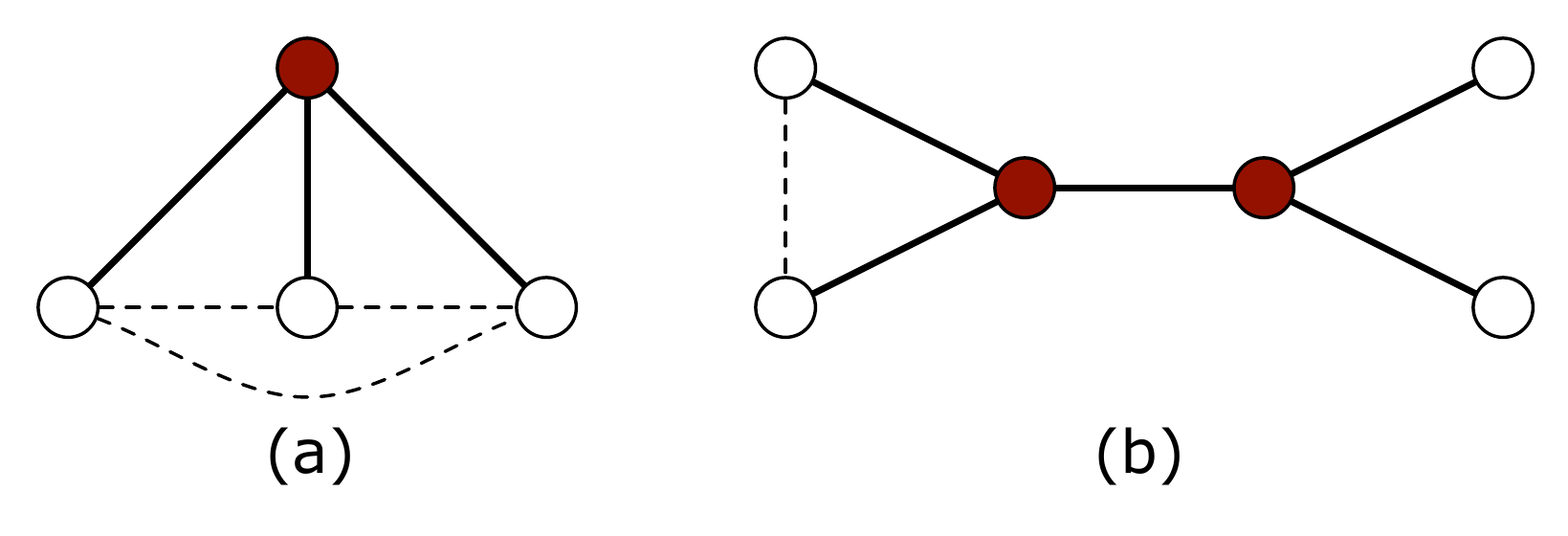}}
\caption{\sl (a) A spanning tree of degree $\opt+1$, and (b) a panning tree of degree $\opt$ (the plain edges are the edges of the trees, and the dotted edges are the remaining edges in the graph, i.e., not in the trees). None of these two trees are FR-trees. Indeed, the dark nodes must be bad while the white nodes must be good, and there are edges between good nodes separated by a bad node.}
\label{fig:FRtree}
\end{figure}

\begin{lemma}\label{lem:plsfrtree}
There is a proof-labeling scheme for FR-trees, using labels on $O(\log n)$ bits, and any silent self-stabilizing FR-tree construction algorithm requires registers of $\Omega(\log n)$ bits. 
\end{lemma}

\begin{proof}
Let $T$ be a FR-tree in $G$. Root $T$ at an arbitrary node $r$. Assume a marking of the nodes of $T$ by ``good'' or ``bad'', witnessing of the fact that $T$ is a FR-tree. The label $L(v)$ of a node $v\in V$ has six fields: 
$
L(v)=(b,d,\ID_r,\ID_f,k,\ell)
$
where each entry is defined as follows. The first entry $b$ is a boolean which is true at node $v$ if and only if $v$ is bad. The second entry $d$ is an integer defined as the distance of $v$ from $r$ in $T$. The third entry, $\ID_r$,  is the identity of the root $r$ of $T$.  Removing the bad nodes from the tree $T$ results in a forest of \emph{fragments} composed of good nodes. Each fragment has a root, defined as the closest node to $r$ in the fragment (distances are computed in $T$). The fourth entry, $\ID_f$, is the identity of the root of the fragment $v$ belongs to, and is called the identity of the fragment. Finally, the integer $k$ is the degree of $T$, and the integer $\ell$ is used to verify that there is a node of degree $k$. 

The verification procedure performs as follows. First, it checks that $T$ is a tree using the aforementioned folklore distance-based method (see Section~\ref{subsec:pls4mdst}). That is, every node $u$ checks that it has the same field $\ID_r$ as its neighbors in $G$. The root $r$ (i.e., the node with $p(r)=\bot$) checks that $\ID(r)=\ID_r$. Moreover, every node $u$ checks that its parent has its distance value equal to $d-1$, and the root $r$ checks that $d=0$. If these tests are passed, then $T$ is necessarily a tree. To verify that this tree is a FR-tree, every node checks that it has the same degree-value $k$ has its neighbors in $G$, and that it has degree at most $k$ in $T$. Moreover, every bad node (i.e., every node with $b=1$) checks that it has degree at least $k-1$ in $T$, and every node of degree $k$ in $T$ checks that it is a bad node. At this point, it remains to check that there is a node with degree $k$ in $T$. This is the role of the last entry $\ell$. A leaf node checks that it satisfies $\ell=1$. Every internal node $u$ checks that its value $\ell$ is equal to the max of the $\ell$-values of its children and its own degree in $T$. The root $r$ verifies that $\ell=k$. It remains to check the last item of Definition~\ref{def:FR-tree}. For this purpose, every node $u$ checks the following with its parent $v=p(u)$. First, if both $u$ and $v$ are good, then they must have the same values of $\ID_f$. If $u$ is bad, or if $u$ is good and $v$ is bad, then $u$ checks that $\ID_f=\ID(u)$. Finally, every good node $u$ checks that each of its neighbors in $G$ is either in the same fragment (it has the same $\ID_f$-value as $u$), or is bad. 

By construction, every node passes all these tests if and only if $T$ is a FR-tree. 

The fact that any silent self-stabilizing FR-tree construction algorithm requires registers of $\Omega(\log n)$ bits is a direct consequence of Theorem~6.1 in~\cite{DGS99}. More specifically, the proof of this latter theorem must be slightly modified in order to manipulate FR-trees. For this purpose, we consider the trees in Fig.~5 in~\cite{DGS99}, and mark all nodes of the ``bottom line'' as bad, and mark all the ``apex nodes'' as good. The resulting tree is a FR-tree. The arguments used in the proof of  Theorem~6.1 in~\cite{DGS99} can then be applied directly. 
\end{proof}

%-------------------------------------------------------------------------
\subsection{The F\"urer and Raghavachari algorithm}
%-------------------------------------------------------------------------

In this section, we recall the algorithm by F\"urer and Raghavachari in~\cite{FR94}  to compute a FR-tree. This algorithm is iterative. Starting from an arbitrary spanning tree $T$, it aims at iteratively reducing the degree of $T$, and proceeds that way until getting to a FR-tree. The code of this (sequential) algorithm can be found in Algorithm~\ref{algo:FR}.  The essential operations performed in this algorithm are Instructions~\ref{ins:12}-\ref{ins:13}, which are aiming at reducing by~1 the degree of a node with maximum degree. See Fig.~\ref{fig:FR} for an illustration of the execution of the Algorithm on different graphs. It was proved in~\cite{FR94} that Algorithm~\ref{algo:FR} produces a FR-tree (and therefore a tree with degree at most $\opt+1$), in time $O(mn\log(n)\alpha(m,n))$ in $n$-node $m$-edge graphs, where $\alpha$ is the inverse Ackerman function. In the remaining of this paper, we shall show that this algorithm can be efficiently implemented in a silent self-stabilizing manner in the state model of distributed computation. That is, our algorithm computes a FR-tree, and stabilizes on FR-trees.

\begin{algorithm}[htb]
\caption{F\"urer and Raghavachari Algorithm (taken from~\cite{FR94})}
\label{algo:FR}
\begin{algorithmic}[1]
\Require{a connected graph $G=(V,E)$}
\State{find a spanning tree $T$ of $G$, and let $k$ be its degree}\label{ins:1}
\Repeat
\State{mark vertices of degree $k$ and $k-1$ as bad, and all other vertices as good}
\State{let $F$ be the set of fragments resulting from removing bad nodes from $T$}\label{ins:4}
\While{there is $e\in E$ between two different fragments \textbf{and} all degree-$k$ vertices are bad}\label{ins:5}
\State{find the bad vertices in the cycle $C$ in $T\cup\{e\}$, and mark them good}\label{ins:6}
\State{update $F$}\label{ins:7}
\EndWhile\label{ins:8}
\State{done $\leftarrow$ there exist no vertices of degree $k$ marked as good}\label{ins:9}
\If{not done} \label{ins:10}
\State{let $w$ be a vertex of degree $k$ marked as good}
\State{find a sequence of improvements which propagate to $w$}\label{ins:12}
\State{update $T$ and $k$ according to this sequence}\label{ins:13}
\EndIf \label{ins:14}
\Until{done}
\State{output $T$ with its vertices marked good or bad}
\end{algorithmic}
\end{algorithm}

\begin{figure}[tb]
\centerline{\includegraphics[width=11cm]{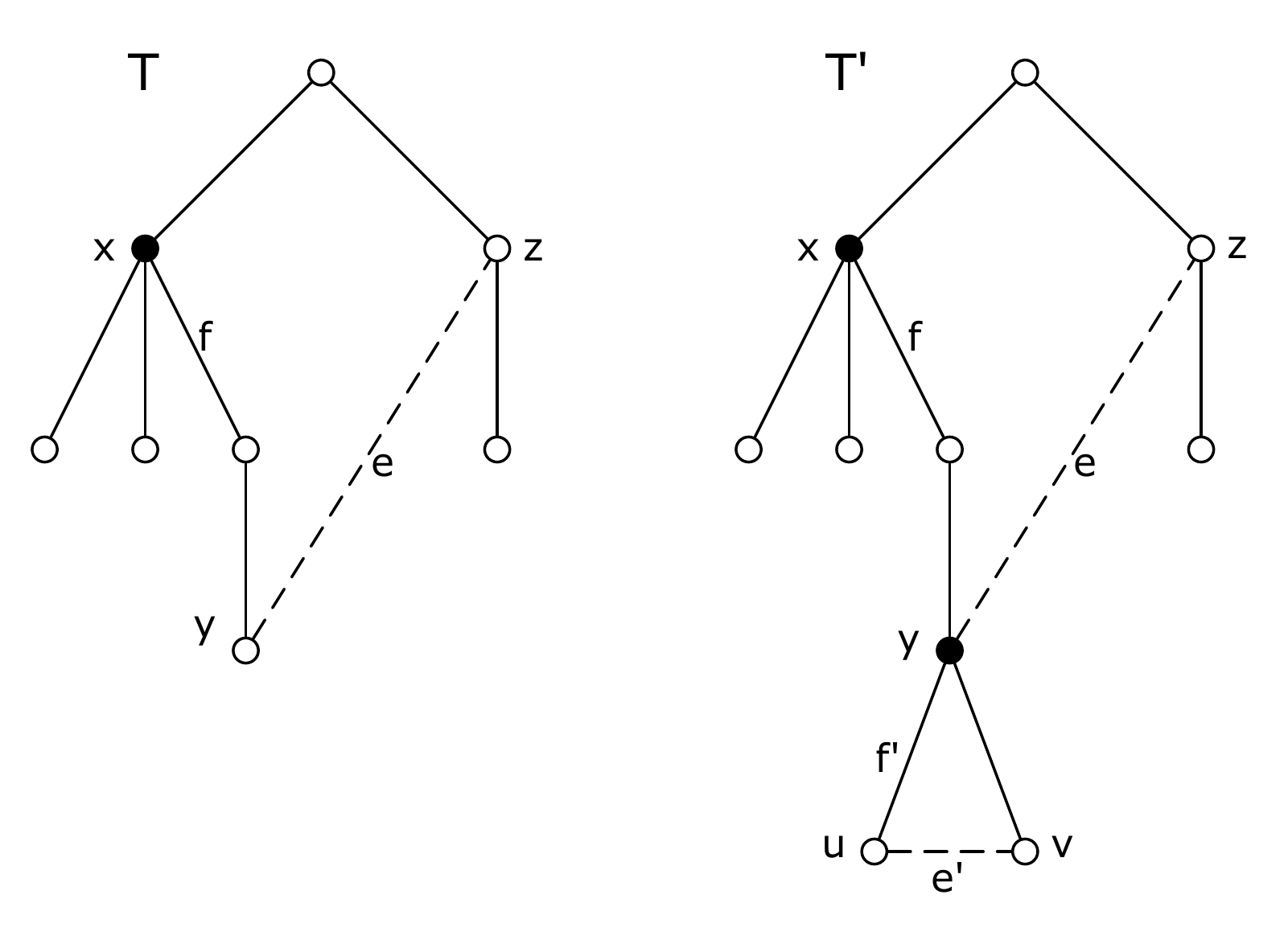}}
\caption{\sl An illustration of the algorithm of F\"urer and Raghavachari. In the left tree $T$, there is only one node of maximum degree $k=4$, and no nodes of degree $k-1=3$. Hence, only node $x$ is marked bad. The non-tree edge $e=\{y,z\}$ connects two fragments. Replacing $f$ by $e$ in $T$ enables to reduce the degree of $x$ by~1, and thus to reduce the degree of the tree $T$ by~1. A slightly more evolved example is depicted on the right hand side. In the right tree $T'$, there is one node of maximum degree $k=4$, and one node of degree $k-1=3$. Hence, both nodes $x$ and $y$ are marked bad. Replacing $f$ by $e$ in $T'$ would decrease the degree of $x$ down to~3, but increase the degree of $y$ up to~4. The algorithm hence identifies first the non-tree edge $e'$ as connecting two different fragments, and turns $y$ to good. This enables to identify the edge $e$ as connecting two different fragments. The sequence of improvement thus consists in: first, replacing $f'$ by $e'$, and then $f$ by $e$. This sequence yields a new tree, with degree~3 instead of~4. Note that the algorithm waits for the full identification of the sequence of improvement before executing these improvements. The while-loop (Instruction~\ref{ins:5}) looks for such a sequence by switching  nodes from bad to good until one node of maximum degree $k$ is switched to good. Then Instruction~\ref{ins:12} constructs the sequence of improvements enabling to decrease the degree of a degree-$k$ node by~1, and Instruction~\ref{ins:13} executes the sequence, and updates the tree accordingly.}
\label{fig:FR}
\end{figure}

%%%%%%%%%%%%%%%%%%%%%%%%%%%%%%%%%%%%%%%%%%%%%
\section{The silent self-stabilizing algorithm for MDST} 
%%%%%%%%%%%%%%%%%%%%%%%%%%%%%%%%%%%%%%%%%%%%%

As we mentioned before, our algorithm is a silent self-stabilizing implementation of Algorithm~\ref{algo:FR}. It is based on three main ingredients closely related to the instruction of this latter algorithm. Instruction~\ref{ins:1} requires to construct a spanning tree, and,  in the course of  actions performed by the algorithm, Instruction~\ref{ins:13} requires to update the current tree (in order either to decrease the number of nodes with maximum degree, or  to reduce the maximum degree). Thus, it is crucial that our spanning tree construction enforces no constraints on the structure of the current tree. Hence, our first ingredient is the design of a silent self-stabilizing spanning tree construction which stabilizes on \emph{any} spanning tree.  

Our second ingredient is also related to Instruction~\ref{ins:13} of Algorithm~\ref{algo:FR}. This instruction requires to perform a sequence of permutations between a tree edge and a non-tree edge in a fundamental cycle of the current tree. The basic mechanism to perform such an permutation is to replace the parent $w=p(v)$ of a node $v$ by another node $w'$. The node $w'$ will of course be such that this replacement does not disconnect the tree (see Fig.~\ref{fig:pruned}). The  permutation is performed by applying a sequence of such local replacements (see Fig.~\ref{fig:pruned-iterated}). In order to perform the replacement loop-free (i.e., to always make sure that the current data structure is a tree), and in a silent manner, we introduce a \emph{redundant} proof-labeling scheme for spanning trees. We prove that this redundant scheme contains enough redundancy for being pruned (i.e., some entries in the scheme can be turned to $\bot$ without hurting the certification). This robustness is then proved to be sufficient for implementing our loop-free silent updates. 

Our third ingredient is related to Instruction~\ref{ins:6} in Algorithm~\ref{algo:FR}. Given a non-tree edge $e=\{u,v\}$ in the current tree $T$, this instruction requires to identify the cycle $C$ formed by $e$ and the path from $u$ to $v$ in $T$. In order to perform this identification locally, we use the informative-labeling scheme for nearest common ancestor (NCA) described in~\cite{AGKR04}. As in any informative-labeling scheme, every node $v$ is provided with a \emph{label} $L(v)$. In NCA-labeling scheme, given the labels $L(u)$ and $L(v)$ of two nodes $u$ and $v$, one can compute the label $L(w)$ of the NCA of $u$ and $v$, $L(w)=\nca(L(u),L(v))$. Using such a scheme, and given $L(u)$, $L(v)$, and $L(w)$, every node can detect whether it belongs to the cycle $C$ or not: $x\in C$ if and only if $\nca(L(x),L(u))=L(x)$ and $\nca(L(x),L(v))=L(w)$, or $\nca(L(x),L(u))=L(w)$ and $\nca(L(x),L(v))=L(x)$. Coming up with a self-stabilizing construction of the labels in the NCA-labeling scheme of~\cite{AGKR04} is easy (actually, \cite{AGKR04} already proposed a distributed algorithm for computing the labels). However, this is not sufficient as, for our algorithm to be silent, we also need to certify the labels. For this purpose, we show how to construct a proof-labeling scheme of the NCA-labeling scheme. 

Having the above three ingredients at hand, we show how to derive a silent self-stabilizing implementation of Algorithm~\ref{algo:FR}. Sections~\ref{subsec:tool1}, \ref{subsec:tool2}, and~\ref{subsec:tool3} respectively describe each of our three ingredient. Finally, Section~\ref{subsec:tool4} describes the implementation of Algorithm~\ref{algo:FR} using these ingredients. 

%-------------------------------------------------------------------------
\subsection{Spanning tree construction and leader election} 
\label{subsec:tool1}
%-------------------------------------------------------------------------

For constructing a spanning tree, we root the tree at the node with minimum identity. This latter node explicitly becomes a \emph{leader}.  Our algorithm is called $\AlgoT$, for ``spanning tree and leader election''. It proceeds essentially in two phases. 

During a first phase, the nodes aim at constructing a spanning forest, i.e., a collection of trees such that each node belongs to at least one tree. During a second phase, the forests are merged into a single spanning tree. Hence, the algorithm uses a boolean predicate $\ErT(v)$ at each node $v$ such that $\ErT(v)$ is true at every node $v$ if and only if the current 1-factor $\{\{v,p(v)\}, v\in V\}$  forms a forest. In other words, this predicate implements a proof-labeling scheme for forests. This scheme is based on the classical distance-to-root method. If an error is detected at a node $v$, that is, if $\ErT(v)$ is true, then the current 1-factor is not a forest. In that case, $v$ aims at restarting every node in the subtree $T_v$ rooted at $v$. (If $v$ has detected the presence of a cycle, then it removes its pointer to its parent, and becomes the root of a subtree).  However, before restarting a node, $v$ launches a procedure whose role  is to ``freeze'' the nodes in the subtree. This freezing process proceeds downward the tree. Freezing nodes allows to avoid creating a cascade of errors, by having, e.g., a node reconnecting itself to a descendent  when it is restarted. When a leaf receives the freezing instruction, it freezes, and then restarts. The restarts then proceed upward the tree, and when a node restarts, it is not allowed to connect to a frozen node. By doing so, all nodes satisfy $\ErT(v)=$ false  after $O(n)$ rounds. Note that some tree in the forest may consist of a single node. To be fully correct, each tree of the forest must actually be rooted  at the node with minimum identity in the tree. This node identity becomes the identity of the tree. 

During the second phase, lasting an additional $O(n)$ rounds, the parent pointers are successively improved in parallel at every node, so that to converge from a forest to a single spanning tree. More precisely, every node $v$  selects as its parent its neighboring node belonging to the tree with smallest identity among all trees spanning the  neighbors of $v$. 

Algorithm $\AlgoT$ is actually tight enough that we can present here a complete implementation of it, and prove the correctness of the implementation. We denote by $N(v)$  the open neighborhood of node $v$ in the graph $G=(V,E)$, i.e., 
$
N(v)=\{u\in V \mid \{u,v\}\in E\}.
$
For a node $v$, we denote by $\ch(v)$  the set of children of $v$ in the current tree. That is, if $T$ is defined by the 1-factor $\{(v,p(v)), v\in V\}$, then we have 
$
\ch(v)=\{u\in V : p(u) = v\}.
$
Our algorithm for constructing a spanning tree uses only the following four variables at each node $v$ (each variable $x$ can also be equal to $\bot$, when it has not yet an assigned value): 
\begin{smallitemize}
\item $\Root_v \in \mathbb{N}$ is the identity of the root of the tree containing $v$ in a current spanning forest;
\item $\p_v \in  \mathbb{N}$ is the identity of the parent $p(v)$ of  $v$ in the tree containing $v$;
\item $\level_v\in \mathbb{N}$ is the distance between $v$ and the root of the tree containing~$v$;
\item $\F_v$ is a boolean ($\F$ stands for ``frozen'').
\end{smallitemize}
Hence, in fact, our algorithm uses only one extra variable compared to the aforementioned classical distance-based proof-labeling scheme for trees. Moreover, this additional variable is just a boolean.  Hence, our algorithm just uses one bit more than the distance-based proof-labeling scheme for trees or forests. 

The algorithm makes use of the following three boolean predicates: 

\begin{eqnarray*}
\TRoot(v)&: & (\Root_v = \ID(v)) \wedge (\p_v = \bot) \wedge (\level_v = 0) \\
\TNode(v)& : &  (\p_v  \in N(v)) \wedge (\Root_v < \ID(v) ) \wedge \big [ \big((\Root_v=\Root_{\p_v}) \wedge (\level_v = \level_{\p_v}+1)\big) \vee (\Root_v>\Root_{\p_v})\big ]  \\
\ErT(v) & : &  \neg\TRoot(v) \wedge   \neg \TNode(v)  
\end{eqnarray*}

Note that $\TRoot(v)=true$ if  $v$ looks locally as a root of a tree, and $\TNode(v)=true$ if  $v$ looks locally as an internal node of a tree. The predicate $\ErT(v)$ is true if an only if node $v$ detects an inconsistency as it cannot be a root nor an internal node. The implementation of Algorithm $\AlgoT$ is specified by the four rules in Algorithm~\ref{alg:LEST}, where the fourth rule, $\R_{\tt Forest}$, uses the following value: 
\[
\Best(v) =
\left\{
  \begin{array}{ll}
   \min\{\ID(u), u\in N(v) \;\text{and}\; \Root_u=\min \{ \Root_w, w\in N(v)\}\} & \text{if } \{u\in N(v) \mid  \F_u=true\}=\emptyset\\
    \ID(v) &\text{otherwise} \\
  \end{array}
\right.\\
\]

\vspace*{-1ex}

\begin{algorithm}[htb]
\caption{\small Implementation of Algorithm $\AlgoT$}
\label{alg:LEST}
\small
\[
\begin{array}{lcllll}
\R_{\tt Error}&\hspace*{-0,3cm}:&\hspace*{-0,2cm}(\ch(v)\neq \emptyset) \wedge \ErT(v) \wedge \big[(\p_v,\F_v)\neq(\bot,true)\big] & \rightarrow& (\Root_v,\p_v,\level_v,\F_v)=(\Root_v,\bot,\level_v,true);\\
\R_{\tt Start}&\hspace*{-0,3cm}:&\hspace*{-0,2cm} (\ch(v)=\emptyset) \wedge \big [ \ErT(v)  \vee (\F_v=true)\big ]& \rightarrow&(\Root_v,\p_v,\level_v,\F_v)=(\ID(v),\bot,0,false); \\
\R_{\tt Freeze}&\hspace*{-0,3cm}:&\hspace*{-0,2cm} \neg  \ErT(v)\wedge (\F_v=false) \wedge (\F_{\p_v}=true) & \rightarrow&(\Root_v,\p_v,\level_v,\F_v)=(\Root_v,\p_v,\level_v,true); \\
\R_{\tt Forest}&\hspace*{-0,3cm}:&\multicolumn{3}{l}{ \hspace*{-0,2cm}\neg  \ErT(v)\wedge \F_v=false) \wedge \big ( (\p_v = \bot) \vee ( (\F_{\p_v}=false ) \big ) \wedge  (\Root_v>\Root_{\Best(v)})} \\
& & \multicolumn{3}{r}{\rightarrow (\Root_v,\p_v,\level_v,\F_v)=(\Root_{\Best(v)},\Best(v),\level_{\Best(v)}+1,false) ;}
\end{array}
\]
\end{algorithm}

 Observe that, at any point in time,  at most guard can be true at $v$. The role of each of the four rules in Algorithm $\AlgoT$ is the following. Essentially, the guard of the first rule is true for node $v$ if $v$ detects that the current 1-factor is not a forest, unless $v$ has not children, or $\p_v=\bot$ and $v$ is frozen. Indeed, in the latter case, it means that $v$ already detected an error, so the same rule must not be applied again. The second rule creates a tree reduced to the single node $v$ as long as $v$ has no children. This corresponds to the aforementioned upward propagation of the restarting nodes. The third rule is the downward propagation of the freezing process. Finally, the fourth rule is the rule insuring that the potentially many trees in the current spanning forest will eventually merge into a single spanning tree, rooted at the node with minimum identity. To be applied, the node $v$ must detect that one of its neighbors $u$ belongs to a tree whose identity $\Root_u$ is smaller than the identity $\Root_v$ of the tree spanning $v$. 

The lemma below is a crucial  tool toward the design of our minimum-degree spanning tree algorithm. It may also have an interest on its own. 

\begin{lemma}
\label{lem:election}
Algorithm $\AlgoT$ is a silent self-stabilizing algorithm that constructs a spanning tree of $G$ rooted at the node with minimum identity, and stabilizes on every spanning tree rooted at the node with minimum identity. It uses $O(\log n)$ bits of memory per node in $n$-node graphs, and stabilizes in $O(n)$ rounds. 
\end{lemma}

\begin{proof}
Let $G=(V,E)$ be the graph in which Algorithm $\AlgoT$ is executed, where each node is provided with a distinct identity in $[0,n^c]$. Let $\Gamma$ be the set of all possible (global) states, that is, $\Gamma$ is the set of all possible collections of individual states of the nodes in $G$. 

We first show that Algorithm \AlgoT\/ does not produce errors among the parent pointers. 

\begin{claim}\label{claim:no-individual-error}
For every node $v$, and every state $\gamma\in\Gamma$, if  $\ErT(v)=$ false in $\gamma$, then, for every state $\gamma'$ reachable from $\gamma$ by Algorithm  $\AlgoT$, we have $\ErT(v)=$ false.
\end{claim}

To establish the claim, first observe that, in the state $\gamma$, only the two rules $\R_{\tt Freeze}$ and $\R_{\tt Forest}$ can be executed at node $v$. The rule $\R_{\tt Freeze}$  has no impacts on the variables $\Root_v, \p_v$, and $\level_v$ used for defining $\ErT(v)$. The rule $\R_{\tt Forest}$ modifies these three variables. However, after having applied $\R_{\tt Forest}$, we have $\Root_{\Best(v)}< \Root_v$ whenever $\Best(v)$ applied $\R_{\tt Forest}$ simultaneously,  otherwise $\Root_{\Best(v)}= \Root_v$ with $\level_v=\level_{\Best(v)}+1$. Therefore the possible actions of $v$ preserves $\ErT(v)=$ false. Now, we prove that the actions of the neighbors of $v$ in the network cannot yield  $\ErT(v)=$ true. Any neighbor $u$ executing $\R_{\tt Error}$ or $\R_{\tt Freeze}$ keeps $\Root_u$ and $\level_u$ unchanged, and hence cannot modify $\ErT(v)$. The parent of $v$ cannot execute $\R_{\tt Start}$. A child $u$ of $v$ executing $\R_{\tt Start}$ results in $u$ disconnecting itself from $v$, and hence does not modify $\ErT(v)$. Finally, consider a neighbor $u$ of $v$ which is neither the parent of $v$, nor one of its children. Such a node $u$ may impact $v$ if it applies $\R_{\tt Forest}$ in order to adopt $v$ as its parent. However, while connecting to $v$, node $u$ precisely sets its variables $\Root_u,\p_v$, and $\level_v$ so that $\ErT(u)=\ErT(v)=$ false. This completes the proof of the claim.

\bigskip

Given any two sets $\Gamma_2\subseteq \Gamma_1 \subseteq \Gamma$, the fact that, given any $\gamma_1\in\Gamma_1$, Algorithm  $\AlgoT$ starting from $\gamma_1$ eventually reaches a state $\gamma_2\in\Gamma_2$, is denoted by
\[
\Gamma1 \rhd \Gamma_2.
\]
Moreover, a set $\Gamma'\subseteq \Gamma$ is said to be \emph{closed} under Algorithm  $\AlgoT$ if, for every $\gamma\in\Gamma'$, every state $\gamma'$ reachable from $\gamma$ by $\AlgoT$ satisfies $\gamma'\in\Gamma'$. 

Let us define $\Gamma_\text{\tt no-cyle}\subseteq \Gamma$  as the set of states $\gamma\in \Gamma$ such that, for every node $v$, we have  
\[
\ErT(v)=false \;\mbox{or}\; (\p_v,\F_v)=(\bot,true).
\] 
Note that an important characteristic of $\Gamma_\text{\tt no-cyle}$ is the absence of cycle in any state $\gamma\in\Gamma_\text{\tt no-cyle}$. Indeed, in such a $\gamma$, any node $v$ with $\ErT(v)=true$ must satisfy $\p_v=\bot$. Thus, such a node cannot be involved in a cycle. On the other hand, in any cycle, at least one node $v$ must satisfy $\ErT(v)=true$ with $\p_v\neq\bot$. 

\begin{claim}\label{claim:no-cylce}
$\Gamma \rhd \Gamma_\text{\tt no-cyle}$ in at most one round, and $\Gamma_\text{\tt no-cyle}$ is closed under Algorithm  $\AlgoT$.
\end{claim}

Let $\gamma\in \Gamma$ be the current state. If $\gamma\notin \Gamma_\text{\tt no-cyle}$, then let $v$ be a node such that $\ErT(v)=true$ and $(\p_v,\F_v)\neq (\bot,true)$. At such a node, either $\R_{\tt Error}$ or $\R_{\tt Start}$  is activatable, depending on whether $v$ has children or not. In both case, the corresponding rule will be activated in the current round as the scheduler is weakly fair.  If $v$ has no children, then, once $v$ has applied $\R_{\tt Start}$, it becomes a root, and thus $\ErT(v)=false$. Instead, if $v$ has at least one  child, then, once $v$ has applied $\R_{\tt Error}$, we have $(\p_v,\F_v)=(\bot,true)$. Therefore, $\Gamma \rhd \Gamma_\text{\tt no-cyle}$ in at most one round. To establish the closeness of $\Gamma_\text{\tt no-cyle}$, we just need to consider the case of nodes applying $\R_{\tt Start}$ because, if $\ErT(v)=false$ then, by Claim~\ref{claim:no-individual-error}, $\ErT(v)$ remains false. So, let $v$ be a node such that $\ErT(v)=true$ and $(\p_v,\F_v)=(\bot,true)$. If such a node $v$ is activated, then it can only apply  $\R_{\tt Start}$, which yields $\ErT(v)=$ false, and thus $\ErT(v)$  will remains false by Claim~\ref{claim:no-individual-error}. This completes the proof of the claim. 

\bigskip

Let $\Gamma_{\tt Forest}$ be the set of states in which, for every node $v\in V$, we have $\ErT(v)=$ false, and $\F_v=$ false. By definition, $\Gamma_{\tt Forest}\subseteq  \Gamma_\text{\tt no-cyle}$. 

\begin{claim}\label{claim:no-error-forest}
$\Gamma_\text{\tt no-cyle}  \rhd \Gamma_{\tt Forest}$ in $O(n)$ rounds, and $\Gamma_{\tt Forest}$ is closed under Algorithm  $\AlgoT$.
\end{claim}

Let $\gamma\in \Gamma_\text{\tt no-cyle}$ be the current state, and let $r_1,\dots,r_k$ be the nodes satisfying that, for $i=1,\dots,k$, $\F_{r_i}=$ true and every ancestor $v$ of $r_i$ satisfies $\F_v=$ false. Let $T_1,\dots,T_k$ be the set of subtrees in  $\gamma$, rooted at $r_1,\dots,r_k$. Such trees are called \emph{icebergs}. For $i=1,\dots,k$, let $h_i(\gamma)$ be the distance in $T_i$ from $r_i$ to the closest descendent $v_i$ in $T_i$ satisfying $\F_{v_i}=$ false, if any. Otherwise, $h_i(\gamma)$ is 1 plus the height of $T_i$, that is 1 plus  the longest distance from $r_i$ to a leaf of $T_i$. Let $\phi_i(\gamma)=2n-h_i(\gamma)$ if $r_i$ has a descendent $v$ with $\F_v=$ false, and $\phi_i(\gamma)=h_i(\gamma)$ otherwise. We define the \emph{potential} $\phi: \Gamma_\text{\tt no-cyle} \to \mathbb{N}$ as follows
\[
\phi(\gamma)=\sum_{i=1}^k \phi_i(\gamma)~.
\]
Observe that, by definition, for every $\gamma\in\Gamma_\text{\tt no-cyle}$, we have $\phi(\gamma) \geq 0$, and $\phi(\gamma)=0$ if and only if $\gamma\in \Gamma_{\tt Forest}$. Also observe that, at each step of Algorithm  $\AlgoT$, the number of icebergs cannot increase. Indeed, since $\gamma\in \Gamma_\text{\tt no-cyle}$ and $\Gamma_\text{\tt no-cyle}$ is closed, the rule $\R_{\tt Error}$ will no more be applied at any node, and none of the three other rules can create icebergs. In fact, the only rule which may decrease the numbers of icebergs is $\R_{\tt Start}$. Since the number of icebergs cannot increase, we analyze the future of each iceberg $T_i$, and hence of each $\phi_i$, separately. 

Assume first that, in $\gamma$,  there are non frozen nodes in $T_i$. Every such non frozen node whose parent is frozen in $\gamma$  is activatable, and thus will be activated during the round since the scheduler is weakly fair. The consequence of this activation is that these nodes become frozen (because of $\R_{\tt Freeze}$). At the end of the round leading from $\gamma$ to $\gamma'$, there are two cases. If there are still non frozen nodes in $T_i$, then $h_i(\gamma') > h_i(\gamma)$, and thus $\phi_i(\gamma') = 2n-h_i(\gamma') < 2n-h_i(\gamma) = \phi_i(\gamma)$. If  all nodes of $T_i$ are frozen at the end of the round, then $\phi_i(\gamma')=h_i(\gamma')+1 < 2n-h_i(\gamma) = \phi_i(\gamma)$ because the height of a tree is at most $n-1$. Assume next that, in $\gamma$,  all nodes of $T_i$ are frozen. In this case, all leaves of $T_i$ are activatable, and thus will be activated during the round. The consequence of this activation is that these nodes disconnect from $T_i$  (because of $\R_{\tt Start}$). As a result, by the end of the round, we get $\phi_i(\gamma')=h_i(\gamma')+1 < h_i(\gamma)+1 = \phi_i(\gamma)$. Hence, $\phi$ is decreasing at each round. Since $\phi_i(\gamma)\leq 2n$, we get that $\phi_i$ becomes null after at most $2n$ rounds starting from $\gamma$. Hence, $\Gamma_\text{\tt no-cyle}  \rhd \Gamma_{\tt Forest}$ in at most $2n$ rounds. 

The fact that $\Gamma_{\tt Forest}$ is closed under Algorithm  $\AlgoT$ is a direct consequence of the fact that, by definition of $\Gamma_{\tt Forest}$, only rules $\R_{\tt Forest}$ can be activated at every node. By Claim~\ref{claim:no-individual-error}, the activation of this rule keeps the state in  $\Gamma_{\tt Forest}$. This completes the proof of Claim~\ref{claim:no-error-forest}. 

\bigskip

Let $\Gamma_{\tt Tree}$ be the set of states in which, for every node $v\in V$, we have $\ErT(v)=$ false, $\F_v=$ false, and the 1-factor $\{(v,\p_v),v\in V\}$ is a tree. By definition, $\Gamma_{\tt tree}\subseteq  \Gamma_\text{\tt forest}$. 

\begin{claim}\label{claim:final}
$\Gamma_\text{\tt forest}  \rhd \Gamma_{\tt tree}$ in $O(D)$ rounds, where $D$ denotes the diameter of the network. Moreover,  $\Gamma_{\tt tree}$ is closed under Algorithm  $\AlgoT$.
\end{claim}

To establish the claim, let $\psi:\Gamma_\text{\tt forest}  \rightarrow \mathbb{N}$ be the potential function defined as:
\[
\psi(\gamma)= \sum_{v\in V} (\Root_v-\ID(r))
\]
where $r$ is the node with the minimum identity among all nodes in $G$. Note that, for every $\gamma\in \Gamma_\text{\tt forest}$, we have $\psi(\gamma)\geq 0$, and $\psi(\gamma)=0$ if and only $\gamma\in  \Gamma_{\tt tree}$. $\Gamma_{\tt tree}$ is closed under Algorithm  $\AlgoT$ simply because no rules are activatable at any node in any state in $\Gamma_{\tt tree}$. For every $\gamma\in \Gamma_\text{\tt forest}$, the only rule that is activatable at every node $v$ is $\R_{\tt Forest}$, and the node $r$ is necessarily the root of a tree $T$ in the forest $F=\{(v,\p_v),v\in V\}$.  All nodes not in $T$ that are neighbors of a node in $T$ are activatable in a round. Each such node $v$ applies  $\R_{\tt Forest}$ during the round since the scheduler is weakly fair. As a result,  $v$ joins $T$ during the round. Hence, after at most $D$ rounds, all nodes have joined the tree rooted at $r$, and hence the state is in $\Gamma_{\tt tree}$ after at most $D$ rounds. This completes the proof of the claim. 

\medskip

The lemma follows by successive application of Claims~\ref{claim:no-cylce}, \ref{claim:no-error-forest}, and~\ref{claim:final}. 
\end{proof}

\noindent\textbf{Remark.} Note that  the way the nodes apply $\R_{\tt Forest}$ in the proof of Claim~\ref{claim:final} does not necessarily yield a BFS tree. Indeed, a node connects to its ``best'' neighbor, defined as the neighbor in the current tree with smallest identity. Of course, the selection of this best neighbor could easily be modified if one would be interested in constructing a BFS tree. However, this is not our point of interest as this tree will be used as a starting point for constructing a minimum-degree spanning tree. In fact, a BFS tree is even generally not a good starting point as BFS trees tend to have high degree in general. 

%-------------------------------------------------------------------------
\subsection{Silent loop-free algorithm for permuting edges in a fundamental cycle}
\label{subsec:tool2} 
%-------------------------------------------------------------------------
 
Given the current spanning tree $T$, Instructions~\ref{ins:12}-\ref{ins:13} of Algorithm~\ref{algo:FR} requires to perform a sequence of updates of the tree, in order to construct another tree $T'$. Implementing one updates boils down to permuting a tree edge with a non-tree edge in a fundamental cycle (i.e., a cycle resulting in adding an edge $\{u,v\}$ between two nodes in a tree. In turn, such a permutation can be performed by a sequence of local operations in which a node permutes the edge connecting it to its parent with an incident non-tree edge. This local operation is called a \emph{switch}. In the following, we start by introducing a \emph{redundant} proof-labeling scheme for trees. Then we show how to use this redundant labeling in order to perform a switch. Finally, we show how to perform a permutation between a tree edge and a non-tree edge in a fundamental cycle, via a sequence of switches. The difficulty is to keep all operations silent, and in particular to never create loop in the 1-factor $\{(v,\p_v),v\in v\}$.  

We start by describing a \emph{redundant} proof-labeling scheme for verifying spanning trees. This labeling scheme  is crucial  with many respects  as it will allow us to modify trees without creating errors due to misinterpreting labels during this modification. In the redundant labeling scheme, the label $L(v)$ of node $v$ is a triple  $(\ID,d,s)$ where, as in the classical distance-based proof-labeling scheme for spanning tree, $\ID$ denotes the identity of the root, and $d$ denotes the distance to this root. The new parameter $s$ denotes the \emph{size} of the subtree rooted at $v$ is the current tree. Hence, in particular, we must have $s_v=1+\sum_{u\in {\footnotesize \ch(v)}}s_u$ where $\ch(v)$ denotes the set of children of $v$ in the tree. It is folklore that, as for the distance-based labeling scheme using the pairs $(\ID,d)$, the pairs $(\ID,s)$ alone provide a proof-labeling scheme for spanning trees. We call this latter labeling the \emph{size-based} labeling scheme. 

\subsubsection{Pruned labeling}

The interesting property of the redundant labeling scheme is that it is flexible enough to let some distance variables  unspecified (e.g., equal to $\bot$). More specifically, let $T$ be a spanning tree of $G$. Assign to each node the due pair $(d,s)$ corresponding to a distance-based labeling scheme, and to a size-based labeling scheme for $T$, with the same identity of the root. We now define a \emph{pruning} of this labeling. To create a pruning of the redundant labeling scheme, one is allowed to replace the entry $d$, or $s$, but not both, by $\bot$, on an arbitrary number of nodes. Not all replacements are however allowed: in addition to forbidding creating pairs $(d,s)=(\bot,\bot)$,  the following two constraints must be satisfied at every node~$v$, where $L(v)=(d,s)$,  $L(p(v))=(d',s')$, and $\hat{L}(\cdot)$ denotes the pruned labeling: 
\begin{smallitemize}
\item[C1:] $\hat{L}(v)=(d,\bot) \Rightarrow \hat{L}(p(v))=(d',\bot) $
\item[C2:] $\hat{L}(v)=(\bot,s) \Rightarrow \hat{L}(p(v))=(d',s') \;\mbox{or}\; \hat{L}(p(v))=(\bot,s') $
\end{smallitemize}

\begin{lemma}
\label{lem:lemmedesprunes}
There exists a verification procedure satisfying the following two properties. {\rm (1)} For any pruning of any legal redundant labeling of any spanning tree $T$, all nodes accept $T$. {\rm (2)}  For any labeling of a non-tree $H$ with triples $(\ID,d,s)$ where $d$ and $s$ can potentially be $\bot$, at least one node rejects $H$.
\end{lemma}

\begin{proof}
We consider the verification procedure described in the table below where ``distance'' stands for checking whether $d_v=d_{p(v)}+1$, and ``size'' stands for checking whether  $s_v=1+\sum_{u\in {\footnotesize \ch(v)}}s_u$, and output yes or no accordingly. (Of course, the presence of a unique root $\ID$ is also checked in all cases, as in the classical distance-based and size-based labeling schemes). 

\begin{center}
\begin{tabular}{cc}
 & Label of $p(v)$ \\
Label of $v$ & \begin{tabular}{c|c|c|c|}
              & $(d',s')$ & $(d',\bot)$ & $(\bot,s')$ \\
 \hline
 $(d,s)$     & \textcolor{blue}{distance} \mbox{and} \textcolor{Green}{size}  & \textcolor{blue}{distance} & \textcolor{Green}{size}   \\
 \hline
 $(d,\bot)$ & \textcolor{red}{no} & \textcolor{blue}{distance} & \textcolor{red}{no}   \\
 \hline
 $(\bot,s)$ & \textcolor{Green}{size} & \textcolor{red}{no} &  \textcolor{Green}{size} \\
 \hline
 \end{tabular}
\end{tabular}
\end{center}
Hence, in particular, if both the label of $v$ and the label of $p(v)$ are intact (i.e., no entries have been turned to $\bot$), then the verification performs at $v$ by checking the distance property between $v$ and $p(v)$, and the size property between $v$ and all nodes in $\ch(v)$. Instead, if the label of $v$ is intact, but the label of $p(v)$ is of the form $(\bot,s)$, then the verification performs at $v$ by checking the size property between $v$ and its children.

Let $T$ be a spanning tree with nodes labeled by pruning a legal redundant labeling. None of the labels can be of the form $(\bot,\bot)$ and the pruning must satisfy C1 and C2. If  $v$ has a label of the form $(d,s)$, and $p(v)$ has a label of either the form $(d',s')$ or the form $(d',\bot)$), then $v$ outputs yes since the distance-based labeling is correct, and, whenever the pruned label of $v$ is of the form $(d,s)$ then, by C2, every child of $v$ has label either $(d'',s'')$ or ($\bot,s'')$. The same holds if  $v$ has a label of the form $(d,\bot)$, and $p(v)$ has a label of the form $(d',\bot)$.  If  $v$ has a label of the form $(d,s)$, but $p(v)$ has a label of the form $(\bot,s')$, then node $v$ checks size. By C1 and C2, all the children of $v$ have labels of the form $(d'',s'')$, or $(\bot,s'')$. Thus $v$ outputs yes since the size-based labeling is correct. Finally, if $v$ has a label of the form $(\bot,s)$, then again it performs check size. By C2, all children of $v$ have labels of the form $(d'',s'')$ or ($\bot,s'')$. Thus $v$ outputs yes since the size-based labeling is correct. As a consequence, the verification accepts the tree $T$, as claimed. 

Conversely, let $H=\{(v,p(v)),v\in V\}$ be the 1-factor induced by the parent pointers, and assume that $H$ is not a spanning tree of $G$. Assume, for the purpose of contradiction, that the verification accepts $H$. As a consequence, all nodes have the same value for the root $\ID$ in theirs labels. Therefore $H$ contains a cycle $C$. The presence on $C$ of a node with pruned label of the form $(d,\bot)$ implies that all nodes of $C$ are of the form $(d,\bot)$, by C1. Hence, all nodes of $C$ are either of the form $(d,s)$ or $(\bot,s)$. In that case, the  size-based labeling detects the cycle, a contradiction. Therefore, we conclude that the verification rejects $H$, as claimed. 
\end{proof}

\subsubsection{The local-switching algorithm}

We are now ready to describe our silent loop-free algorithm for the switching task, called $\AlgoA$. The switching task is specified by a rooted spanning tree $T$, a node $v$ distinct from the root, with parent $p(v)=w$, and a node $w'\in N(v)$ different from any descendent of $v$. The objective is for $v$ to assign $w'$ as its new parent, resulting in a new rooted spanning tree $T'$. The algorithm proceeds in three phases.  (It is here assumed that the nodes of the spanning tree $T$ are initially correctly labeled by the redundant proof-labeling scheme for trees). See Figure~\ref{fig:pruned} for a graphical representation of the algorithm.

\begin{figure}[tb]
\centerline{\includegraphics[width=11cm]{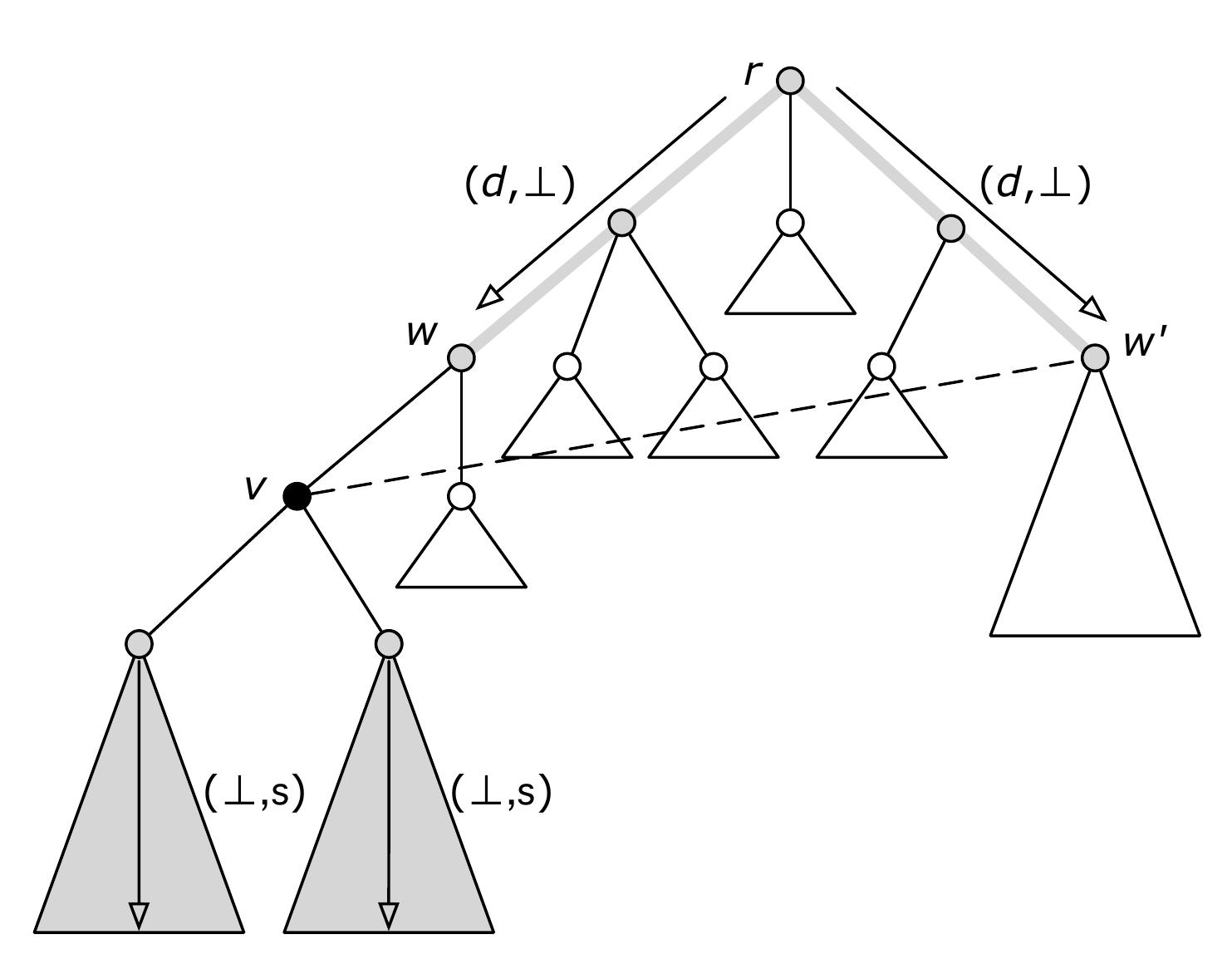}}
\caption{\sl Node $v$ switches parent from $p(v)=w$ to $p(v)=w'$.}
\label{fig:pruned}
\end{figure}

\begin{enumerate}
\item Pruning phase. Node $v$ initiated three ``waves'' of updates for the redundant labeling, resulting in a pruned labeling. First, a signal goes upward from $w$ to the root $r$, while another signal goes upward from $w'$ to $r$. Once $r$ receives any of these two signals, it initiates  a wave of label-updates downward along the path from where it got the signal: every node $u$ on the path $P_{r,w}$ from $r$ to $w$ (respectively, on the path $P_{r,w'}$ from $r$ to $w'$) receiving this updating wave modifies its label from $(d,s)$ to $(d,\bot)$. The third wave of updates goes downward the subtree $T_v$ rooted at $v$, starting from every child of $v$. The nodes in this subtree successively prune their label from $(d,s)$ to $(\bot,s)$, once their parents have done so. Hence, all these updates result in a legal pruning of the original redundant labeling of $T$. 

\item Switch phase. Once the labels of $w$ and $w'$ have both be pruned to a label of the form $(d,\bot)$, and all children of $v$ have their labels turned to $(\bot,s)$, node $v$ sets it parent to $w'$, i.e., we now have $p(v)=w'$. Simultaneously, node $v$ updates its distance to 1 plus the distance of $w'$ to $r$. 

\item Relabeling phase. The former parent $w$ of $v$ then recomputes the size of its subtree, by adding the sizes of all its children (those children have kept the size of their subtrees in their label). Every node along the path $P_{r,w}$ from the former parent $w$ of $v$ to $r$ proceeds the same successively, upward, up to $r$. Similarly, the new parent $w'$ of $v$ recomputes the size of its subtree, by adding the sizes of all its children (including $v$, which kept the size of its subtree in its label). Then every node along the path $P_{r,w'}$ from $w'$  to $r$ proceeds the same successively, upward, up to $r$. Once $v$ has changed its parent to $w'$, every node in the subtree $T_v$ of $v$ recomputes its distance to the root, successively downward, down to the leaves. 
\end{enumerate}

\begin{lemma}
\label{lem:loop-free}
Algorithm $\AlgoA$  executes the switch task in $O(n)$ rounds. It is silent, loop-free, and uses $O(\log n)$ bits of memory at each node. 
\end{lemma}

\begin{proof}
The modifications of the labels performed during the pruning phase results in a legal pruned labeling of the tree~$T$. Indeed, the updates performed downward along the path $P_{r,w}$, as well as those performed along the path $P_{r,w'}$, satisfy C1. Similarly, the updates performed downward in the subtree $T_v$ of $v$ satisfy C2. During the switch phase, the label of $v$ is modified. However, its distance remains consistent with its new parent $w'$, and its size is not modified  (and thus remains consistent with its children). Finally, the reconstruction of the labels performed during the relabeling  phase always guarantees C1 and C2. So overall, Algorithm $\AlgoA$  is loop-free, and silent. 

The pruning phase and the relabeling phase take a number of rounds at most 
\begin{equation}\label{eq:roundswitch}
O(\max\{|T_v|,|P_{r,w}|,|P_{r,w'}|\}).
\end{equation} 
Indeed, as the scheduler is weakly fair, the pruning  and the relabeling  performs essentially concurrently in $T_v$, $P_{r,w}$, and $P_{r,w'}$. Moreover, the switch phase performs in $O(1)$ rounds. Hence, in total, Algorithm $\AlgoA$ completes in $O(n)$ rounds. 
\end{proof}

We provide below an implementation of Algorithm $\AlgoA$. In addition to the variables used for the implementation of Algorithm~ $\AlgoT$, we use the following variables. 
\begin{smallitemize}
\item $\size_v \in \mathbb{N}$ is the number of nodes in a sub-tree $T_v$ rooted at $v$, including $v$.
\item $\switch_v \in \mathbb{N}$ is the identity of the new parent of $v$.
\item $\up_v$ is a boolean indicating the direction of the ``waves''.  
\end{smallitemize}
The implementation also uses the function $\Tsize$, defined by $\Tsize(v)= 1+\sum_{u\in \footnotesize\ch(v)}\size_u$ at every node~$v$. The error  predicate used in $\AlgoA$ is an extension of the one used in $\AlgoT$: 
\begin{eqnarray*}
\TRoot(v)&: & (\Root_v = \ID(v)) \wedge (\p_v = \bot) \wedge (\level_v=0) \wedge  (\size_v= 1+\mbox{$\sum_{u\in \footnotesize\ch(v)}\size_u$}) \\
\TNode(v)& : &  (\p_v  \in N(v)) \wedge (\Root_v=\Root_{\p_v}) \wedge [(\level_v = \bot)\vee (\level_{\p_v}=\bot) \vee  (\level_v=\level_{\p_v}+1) ]  \\
&&\wedge  [(\size_v = \bot)\vee (\exists u\in \ch(v):\size_u=\bot) \vee  (\size_v= 1+\mbox{$\sum_{u\in \footnotesize\ch(v)}\size_u$}) ] \\
&&\wedge \; (\level_v,\size_v)\neq(\bot,\bot) \\
&&\wedge \;[(\level_v,\size_v)\neq(d,\bot) \vee (\level_{\p_v},\size_{\p_v})\neq(d,s)]\\
&&\wedge \;[(\level_v,\size_v)\neq(d,\bot) \vee (\level_{\p_v},\size_{\p_v})\neq(\bot,s)]\\
&&\wedge \; [(\level_v,\size_v)\neq(\bot,s) \vee (\level_{\p_v},\size_{\p_v})\neq(d,s)]\\
\ErL(v) & : & \big( \neg\TRoot(v) \wedge   \neg \TNode(v) \big) \vee \big ((\up_v=false) \wedge (\up_{\p_v}=true)\big)
\end{eqnarray*}
Finally, the algorithm makes use of the following boolean predicates on the local variables of every node~$v$: 
\begin{eqnarray*}
\TDown(v,X)&:& (X_v=\bot) \wedge (X_{\p_v}\neq \bot)\\
\TUp(v,X)&:& (X_v=\bot) \wedge (\forall u\in \ch(v), X_u\neq \bot)\\
\UpWave(v)&:& (\exists u \in \ch(v)\mid \switch_u\neq \bot) \text{ or }  (\exists u \in N(v)\mid \switch_u=v ) \\
& & \text{ or }  (\exists u \in \ch(v)\mid \up_u=true ) 
\end{eqnarray*}
The implementation of Algorithm $\AlgoA$ is described in Algorithm~\ref{alg:algoA}. In this implementation, ``Alarm'' stands for the procedure that is launched in case an error is detected. Roughly, the algorithm returns to $\AlgoT$, and resets all variables specific to $\AlgoA$ ($\size_v$, $\switch_v$, and $\up_v$ to $\bot$).

\begin{algorithm}[htb]
\caption{Implementation of algorithm $\AlgoA$}
\label{alg:algoA}
{\small
\[
\begin{array}{lcllll}
\R_{\tt Error}\hspace*{-0,3cm}&:&\ErL(v)  & \rightarrow& \text{Alarm};\\
\R_{\tt CleanD}\hspace*{-0,3cm}&:&\hspace*{-0,3cm}\neg \ErL(v)\wedge (\level_v\neq\bot) \wedge (\switch_v=\bot) \wedge[(\switch_{\p_w}\neq\bot)\vee(\level_{\p_v}=\bot)]  & \rightarrow& \level_v:=\bot;\\
\R_{\tt Up}\hspace*{-0,3cm}&:&\hspace*{-0,3cm}\neg \ErL(v) \wedge (\size_v\neq\bot) \wedge (\up=false) \wedge \UpWave(v)& \rightarrow& \up:=true;\\
\R_{\tt CleanS}\hspace*{-0,3cm}&:&\hspace*{-0,3cm}\neg \ErL(v)\wedge(\size_v\neq \bot) \wedge (\up_v=true) \wedge [\TRoot(v) \vee (\size_{\p_v}=\bot)&\rightarrow &\size_v:=\bot;\\
&&& &\up_v:=false;\\
\R_{\tt Switch}\hspace*{-0,3cm}&:&\hspace*{-0,3cm}\neg \ErL(v)\wedge(\switch_v\neq\bot)\wedge (\size_{\p_v}=\bot) \wedge  (\size_{\switch_v}=\bot)\wedge \TUp(v,\level)&\rightarrow &\p_v:=\switch_v;\\
&&& &\level_v:=\level_{\switch_v}+1;\\
&&& &\switch_v:=\bot;\\
\R_{\tt Dist}\hspace*{-0,3cm}&:&\hspace*{-0,3cm}\neg \ErL(v)\wedge \TDown(v,L) \wedge (\forall u\in \ch(v),\level_u=\bot) \wedge (\switch_{\p_v}=\bot) & \rightarrow& \level_v:=\level_{\p_v}+1;\\
\R_{\tt Size}\hspace*{-0,3cm}&:&\hspace*{-0,3cm}\neg \ErL(v)\wedge \TUp(v,\size) & \rightarrow& \size_v:=\Tsize(v);\\
\end{array}
\]
}
\end{algorithm}

\subsubsection{The permutation algorithm}

One basic ingredient in Algorithm~\ref{algo:FR} is to remove an edge pending down from a node of high degree, and to reconnect the disconnect spanning tree to another node (of smaller degree) somewhere else in the tree. Let $P$ a simple path in the tree with a non-tree edge $e$ connecting its two extremities. Let $C$ be the fundamental cycle created by adding $e$ to $P$, and let $f\neq e$ be another edge of $C$ (see Fig.~\ref{fig:pruned-iterated}). We aim at replacing $f$ by $e$ in $T$. For this purpose, we perform a sequence of local switches. 

Let $e=\{v,w'\}$ and assume, w.l.o.g., that $f$ is on the simple path connecting $v$ and the nearest common ancestor of $v$ and $w'$.  (See Fig.~\ref{fig:pruned-iterated}). Then let $w$ be the upmost extremity of $f$, and let $w=u_0,u_1,\dots,u_{k-1},u_k=v$ be the simple path in $T$ from $w$ to $v$.  In order to replace $f$ by $e$ in $T$, it is sufficient to perform a sequence of local switches: between $\{v,u_{k-1}\}$ and $e=\{v,w'\}$, then between $\{u_{k-1},u_{k-2}\}$ and $\{u_{k-1},v\}$, and so on by performing successively the switches between $\{u_{i},u_{i-1}\}$ and $\{u_{i},u_{i+1}\}$ until the last switch $\{u_{1},w\}$ and $\{u_{1},u_{2}\}$ which eventually remove $f$, and completes the operation.

\begin{figure}[tb]
\centerline{\includegraphics[width=11cm]{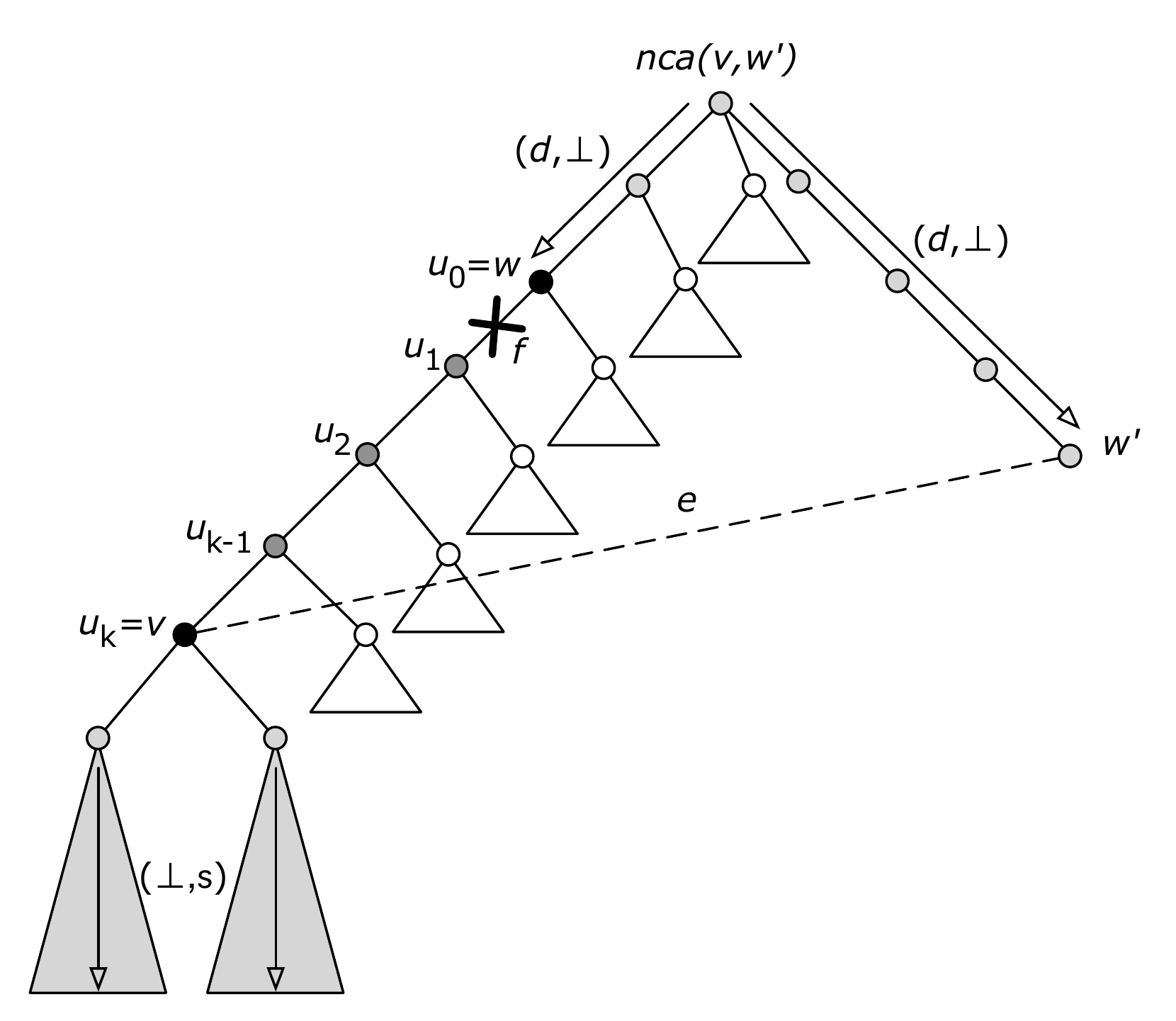}}
\caption{\sl The edge $f$ incident to $w$ on the path to $v$ should be permuted with the edge $e=\{v,w'\}$.}
\label{fig:pruned-iterated}
\end{figure}

The algorithm executing the sequence of switches as above is called $\AlgoP$. Interestingly, although Algorithm $\AlgoP$ may execute an arbitrary large sequence of Algorithms $\AlgoA$, its complexity remains $O(n)$ rounds. Indeed, we first observe that the pruning corresponding to changing labels of the form $(d,s)$ into ones of the form $(\bot,s)$ occurs in different subtrees of nodes $u_0,u_1,\dots,u_k$. Since the total size of all these subtrees cannot exceed $n$, these trees will be pruned and relabeled in $O(n)$ rounds. On the other hand, the pruning for $(d,s)$ to $(d,\bot)$ can be performed only once during Algorithm $\AlgoP$, with the relabeling starting upward only after $f$ has eventually be removed. 

The following result summarizes the content of the section.  

\begin{lemma}
\label{lem:loop-free2}
Algorithm $\AlgoP$  executes the permutation task in $O(n)$ rounds. It is silent, loop-free, and uses $O(\log n)$ bits of memory at each node. 
\end{lemma}

%-------------------------------------------------------------------------
\subsection{Silent NCA-labeling scheme construction} 
\label{subsec:tool3} 
%-------------------------------------------------------------------------

Our third tool for implementing the algorithm by F\"urer and Raghavachari in a silent self-stabilizing manner using $O(\log n)$ bits of memory per node is the implementation of an informative-labeling scheme for nearest common ancestor (NCA) using   $O(\log n)$ bits of memory per node. We use the scheme described in~\cite{AGKR04}. 

This scheme is based on the heavy-path decomposition of the rooted tree similar to the one used in~\cite{HT84}. Recall that, in this decomposition, each node of $T$ is classified as either heavy or light. The root is light. For each internal node $v$, let $w$ be a child of $v$ such that  $|T_w| = \max\{ |T_u| | u\in\ch(v)\}$,  and classify $w$ as heavy. Each of the remaining children of $v$ are classified as light. An edge connecting a heavy node to its parent is a heavy edge. All the other edges are light. A heavy path is a maximal path composed of heavy edges. Finally, given a node $v$, $\apex(v)$ denotes the nearest ancestor of $v$ which is light (possibly $v$ itself if $v$ is light). Note that the number of light edges traversed by the shortest path from the root to any node $v$ is at most $O(\log n)$. Hence, the heavy-path decomposition provides a NCA-labeling scheme using $O(\log^2n)$ bits labels (see~\cite{P05}). In this scheme, the label $L(v)$ of node $v$ is defined as 
\[
L(v)=(L(p(\apex(v))),\ID(\apex(v)),\dist(v,\apex(v))).
\]
Given the labels $L(u)$ and $L(v)$ of two nodes $u$ and $v$, the label of the nearest common ancestor $\nca(u,v)$ of $u$ and $v$ is essentially  the longest common prefix of $L(u)$ and $L(v)$. This latter scheme uses labels on $O(\log^2n)$ bits mainly because it stores  node identities in the labels, and encodes distances.  To reduce the label size down to $O(\log n)$ bits, the scheme in \cite{AGKR04} uses \emph{alphabetic codes}~\cite{GM59}, to replace both the distances along the heavy path, and the identities of the apexes. 

Let $S=(n_1,n_2,\dots,n_k)$ be a sequence of positive integers. For $i=1,\dots,k$, let $s_i=\sum_{j=1}^{i}n_j$ and $\nu_i=\lfloor \log_2 n_i\rfloor$. Let $s_0=0$. The code of $S$ uses, for each interval $[s_{i-1},s_i)$, the smallest multiple $z_i$ of $2^{\nu_i}$ which is not smaller than  $s_{i-1}$. That is, 
\[
z_i=\left\{\begin{array}{ll}
s_{i-1} & \mbox{if $s_i \equiv 0 \pmod {2^{\nu_i}}$} \\
s_{i-1}+2^{\nu_i}-(s_{i-1} \bmod 2^{\nu_i}) & \mbox{otherwise}
\end{array}\right.
\]
Let $w_i$ be the binary word on $\lceil \log_2 n\rceil$ bits encoding $z_i$, potentially with an additional  sequence of leading $0$'s. Since $z_i$ is a multiple of $2^{\nu_i}$, the word $w_i$ has at least $\nu_i$ least significant bits equal to~$0$. Thus, let $c_i$ be the binary word resulting from removing the $\nu_i$ least significant bits from $w_i$. The sequence $(c_1,c_2,\dots,c_k)$ is the alphabetic codes of $S$. An important property of this alphabetic code is that $0.c_1 < 0.c_2 < \dots < 0.c_k$~\cite{GM59}.

Let $\lightsize(v)$ be the light-size of $v$, defined as $\lightsize(v)=|T_v|-|T_w|$ where $w$ is the heavy child of $v$. Let $v_1,v_2,\dots,v_k$ be a heavy path, where $v_1$ is the apex. Let $(x_1,x_2,\dots,x_k)$ be the sequence of integers defined by $x_i=\lightsize(v_i)$. Let $(c_1,c_2,\dots,c_k)$ be the alphabetic code of this sequence. The string $c_i$ is called the heavy-label  of $v_i$, denoted by $\mbox{hlabel}(v_i)$. The labeling stores $\mbox{hlabel}(v_i)$ in $L(v_i)$ instead of $\dist(v_i,v_1)$. Similarly, let  $u_1,u_2,\dots,u_k$ be the light children of a heavy node $v$, ordered such that $\ID(u_1)<\ID(u_2)<\dots<\ID(u_k)$.  Let $(y_1,y_2,\dots,y_k)$ be the sequence of integers defined by $y_i=|T_{v_i}|$, and let $(b_1,b_2,\dots,b_k)$ be the alphabetic code of this sequence. The string $b_i$ is called the light-label  of $u_i$, denoted by $\mbox{llabel}(u_i)$. The labeling stores $\mbox{llabel}(u_i)$ in $L(u_i)$ instead of $\ID(v_i)$. The overall label of a node is thus: 
\[
L(v)=(L(p(\apex(v))),\mbox{llabel}(\apex(v)),\mbox{hlabel}(v)).
\]
Using the properties of alphabetic code, one can show that the labels produced this way are on $O(\log n)$ bits (see~\cite{AGKR04}). As for the scheme on $O(\log^2n)$ bits, given the labels $L(u)$ and $L(v)$ of two nodes $u$ and $v$, the label of the nearest common ancestor $\nca(u,v)$ of $u$ and $v$ is essentially  the longest common prefix of $L(u)$ and $L(v)$.

The NCA-labeling scheme above can be constructed distributedly using $O(\log n)$ bits of memory per node, as shown in~\cite{AGKR04}. This construction can easily be made self-stabilizing. To be silent, every node $v$ must however stores additional information to prove the correctness of the NCA-labeling. 

\begin{lemma}
\label{lem:proof-of-informative}
There is a proof-labeling scheme for the NCA-labeling of~\cite{AGKR04} which uses labels on $O(\log n)$ bits per node.  The construction of the NCA-labeling, as well as of its proof can be done in $O(n)$ rounds. 
\end{lemma}

\begin{proof}
All nodes stores $\lceil \log_2 n\rceil$ where $n$ is the size of the tree $T$. Each node $v$ additionally stores $|T_v|$ and $\lightsize(v)$. Let $v_1,v_2,\dots,v_k$ be a heavy path, where $v_1$ is the apex. We define $\mbox{ssize}(v_i)=\sum_{j=1}^{i-1}\lightsize(v_j)$, and let $\mbox{hchild}(v_i)=\ID(v_{i+1})$ ($\mbox{hchild}(v_k)=\bot)$. Observe that verifying the consistencies of all these variables at each node can be easily achieved by checking the registers of the neighboring nodes in $T$. 

To check the correctness of the NCA-labeling, each heavy node $v$ first checks that its label $L(v)$ is identical to the one of its parent, but the last field, $\mbox{hlabel}(v)$. To check the correctness of this latter field, node $v$ just recomputes its heavy label using all the aforementioned additional information. (In particular, given  $\mbox{ssize}(v)$ and $\lightsize(v)$, node $v$ can compute the partial sums $s_{i-1}$ and $s_i$ used to set up the alphabetic code). 

Checking the correctness of the NCA-labeling at a light node $v$ (i.e., at an apex) is slightly more complex because computing $\mbox{llabel}(v)$ requires to know information that are not necessary available on the neighboring nodes. These information are however available at distance~2, on the sibling nodes. Therefore, the parent $p(v)$ of $v$ can do the checking for $v$. More precisely, a node $u$ having $v_1,v_2,\dots,v_k$ as light children (ordered by identities) recomputes sequentially the light label $\mbox{llabel}(v_i)$ of each of its light child $v_i$, and rises an alarm if one of these light labels was not appearing in the label of $L(v_i)$. Beside this, node $v_i$ just checks that  its label $L(v_i)$ is a suffix of $L(u)$. 
\end{proof}

%-------------------------------------------------------------------------
\subsection{Silent FR-tree construction} 
\label{subsec:tool4}
%-------------------------------------------------------------------------

We have now all the ingredients to describe our silent self-stabilizing algorithm for the minimum-degree spanning tree problem using $O(\log n)$ bits of memory per node, and stabilizing in a polynomial number of rounds. 
Recall that these ingredients are a silent self-stabilizing algorithm for constructing an arbitrary spanning tree rooted at the node with minimum identity (cf. Section~\ref{subsec:tool1}), a silent loop-free algorithm for switching between two spanning trees (cf. Section~\ref{subsec:tool2}), and a silent construction of a labeling scheme for nearest-common ancestor  (cf. Section~\ref{subsec:tool3}). All these algorithms use registers of $O(\log n)$ bits at each node, and stabilize in $O(n)$ rounds. Using these ingredients, our algorithm for constructing a minimum-degree spanning tree is a distributed self-stabilizing implementation of the sequential algorithm by F\"urer and Raghavachari~\cite{FR94} (see Algorithm~\ref{algo:FR}). This implementation presents no more difficulties given our three precious ingredients. 

The first instruction of Algorithm~\ref{algo:FR}, i.e., constructing a spanning tree $T$ of the current graph~$G$, is achieved using our silent self-stabilizing algorithm, $\AlgoT$,  for constructing an arbitrary spanning tree. At any point in time during the execution of the algorithm, computing the maximum degree $k$ of the nodes in the current tree is achieved by a convergecast to the root, and computing the status of the nodes (good or bad) as well as the identities of the different fragments is achieved by a divergecast from the root (each fragment in the tree is identified by the identity of the fragment's node which is closest to the root). The test of Instruction~\ref{ins:9} is achieved by a convergecast in the tree which collects the number of degree-$k$ good nodes at the root. The root  then decides that the algorithm terminates if this number is zero. By the definition of a FR-tree (cf. Definition~\ref{def:FR-tree}), when the algorithm terminates, its degree is at most $\opt+1$. So, from now on, we describe the self-stabilizing implementation of Algorithm~\ref{algo:FR} by focussing on the search for improvements in the while-loop (cf. Instructions~\ref{ins:5}-\ref{ins:8}), and the update of the spanning tree (cf. Instructions~\ref{ins:10}-\ref{ins:14}). 

\subsubsection{Search for a sequence of improvements} 

The search for improvements is  using the self-stabilizing implementation of the NCA-labeling scheme of~\cite{AGKR04} described in Section~\ref{subsec:tool3}. Any cycle composed of a path in a current tree $T$ plus a non-tree edge $e=\{u,v\}$ connecting the two extremities $u$ and $v$ of that path is encoded as $(L(u),L(v),L(w))$ where $w=\nca(u,v)$ and $L(\cdot)$ is the node labeling of ~\cite{AGKR04}. At any execution of the while-loop, every good node $u$ identifies its non-tree edges incident to good nodes in different fragments, and chooses one of them arbitrarily (e.g., by choosing the edge $\{u,v\}$ such that $\nca(u,v)$ is closest to the root, which can be done based solely on the label of $\nca(u,v)$, which is in turn computable based only on $L(u)$ and $L(v)$). Once this edge has been identified, the good node $u$ stores the cycle corresponding to that edge. The selection of a unique edge among all non-tree edges selected by the nodes is performed via a convergecast where each internal node in the tree proceeds as follows. Given the cycles 
\[
\big (L(u_1),L(v_1),L(w_1)\big ),\dots, \big (L(u_d),L(v_d),L(w_d)\big)
\]
 of the $d$ children of a node $u$, and given the cycle  $(L(u_0),L(v_0),L(w_0))$ of $u$ itself, node $u$ selects the one, say $(L(u_i),L(v_i),L(w_i))$, such that $w_i$ is closest to the root. 

Once the root has selected a cycle $(L(u),L(v),L(w))$, a divergecast starts whose role is to switch the bad nodes on that cycle to good. Note that not only the NCA-labeling scheme enables to compute the nearest common ancestor of any pair of nodes, it also enables to perform routing. Thus, all bad nodes on the path from $w$ to $u$, as well as on the path from $w$ to $v$, switch to good. In addition, those nodes store the pair $(L(u),L(v))$, since the corresponding edge $\{u,v\}$ will potentially  be used from decreasing their degrees. 

Finally, a convergecast is performed in order to detect whether at least one bad node of degree~$k$ has switched to good. If this is the case, the algorithm updates the spanning tree as described hereafter. Otherwise, another loop is initiated. This proceeds until either one bad node of degree~$k$ has switched to good, or no improving edges have been detected (i.e., no cycles have been received by the root after the convergecast of cycles). In this latter case, as said before, the root terminates. Hence, we now focus on the way to update the current spanning tree in order to decrease the degree of at least one node of degree $k$. 

\subsubsection{Updating the spanning tree}

We are here in the situation in which at least one node of  degree $k$ has switched from bad to good. In case there are more than one such nodes,  the root can select one of them. So, from now on, we are concerned with improving the degree of a unique node of degree $k$.  Recall from the above that any bad node $x$ that switches to good keeps in its register a pair $(L(u),L(v))$ of labels corresponding to an edge $\{u,v\}$ which can be used for decreasing the degree of $x$. More specifically, assume, w.l.o.g., that $u$ is a descendent of $x$ (by construction, one of the two nodes $u$ and $v$ is a descendent of $x$). Then the tree 
\[
T' = T \setminus \{x,y\} \cup \{u,v\}
\]
where $y$ is the child of $x$ on the shortest path from $x$ to $u$ is another spanning tree of $G$ in which the degree of $x$ has decreased by~1. If $\deg_T(u)<k-1$ and $\deg_T(v)<k-1$, then the change from $T$ to $T'$ can be implemented using the algorithm $\AlgoP$  described in Section~\ref{subsec:tool2}. However, if one of these two nodes (or both)  has degree $k-1$, then this change from $T$ to $T'$ requires first to decrease the degree of those nodes with degree $k-1$. This phenomenon can actually repeats for $u$ and/or for $v$, and for the nodes which could help reducing their degrees. We are thus potentially  facing a sequence of improvements, which can actually be viewed as a binary tree $B$ rooted at $x$, where the at most two children of a node $z$ are the nodes yielding a degree-improvement for $z$. (See~\cite{FR94} for the absence of loops in the sequence of improvements). 

The binary tree $B$ is actually  stored distributedly at each node $z$ to be improved, since every such node precisely stores in its register the pair of labels corresponding to the extremities of the improving edge for $z$. Therefore, the sequence of improvements is executed in a postorder manner, starting from one leaf of $B$ until one can eventually improve the degree of $x$. (The postorder is provided directly by the structure of the distributed storage of $B$ at the nodes).  Each improvement in this sequence of improvements can be implemented using the algorithm $\AlgoP$ described in Section~\ref{subsec:tool2}.

Once the degree of the identified node of degree $k$ has been reduced from $k$ to $k-1$, the algorithm starts a new iteration of the repeat-loop. 

\subsection{Space and time complexity}

We are now ready to state our main result: 

\begin{theorem}
\label{theo:main}
There exists a silent self-stabilizing algorithm which constructs and stabilizes on FR-trees, a subclass of spanning trees with degree at most $\opt+1$. The algorithm uses an optimal memory of $O(\log n)$ bits per node, converges in a polynomial number of rounds, and performs polynomial-time computations at each node.  
\end{theorem}

\begin{proof}
We prove that the algorithm described in the previous sections satisfies the statement of the theorem. The correctness of the algorithm, including the fact that it constructs and stabilizes on a minimum-degree spanning tree with degree at most $\opt+1$, follows directly from the correctness of Algorithm~\ref{algo:FR} established by F\"urer and Raghavachari in~\cite{FR94}, and from the correctness of each of the three ingredients used to implement it, established in Lemmas~\ref{lem:election}, \ref{lem:loop-free2}, and~\ref{lem:proof-of-informative}. The fact that it uses $O(\log n)$ bits of memory per node also follows from Lemmas~\ref{lem:election}, \ref{lem:loop-free2}, and~\ref{lem:proof-of-informative}, and from the fact that, apart from the variables used for the spanning tree construction, for the switches and permutations, and for the NCA labeling, the algorithm uses only variables on $O(\log n)$ bits (to store the maximum degree $k$ of the tree, the nature good or bad of the nodes, etc.). The fact that $\Omega(\log n)$ bits are requires follows from Lemma~\ref{lem:plsfrtree}. The amount of individual computation performed at each node is polynomial, including the NCA-labeling (cf.~\cite{AGKR04}). Even if the number of rounds performed for the spanning tree construction, for the permutations, and for the NCA labeling, amounts for $O(n)$ each, it remains that the overall number of rounds of the algorithm might be larger. Nevertheless, the number of iterations of the repeat-loop in  Algorithm~\ref{algo:FR} cannot exceed~$n^2$. Indeed, the number of nodes with maximum degree $k$ decreases by at least one at each iteration, and $k$ can take at most $n-2$ different values. Thus, overall, the number of rounds of our implementation is polynomial\footnote{Cf. Section~\ref{sec:conclusion} for a discussion about the round complexity of constructing minimum-degree spanning trees.}.
\end{proof}

%%%%%%%%%%%%%%%%%%%%%%%%%%%%%%%%%%%%%%%%%%%%%
\section{Conclusion} 
\label{sec:conclusion}
%%%%%%%%%%%%%%%%%%%%%%%%%%%%%%%%%%%%%%%%%%%%%

In this paper, we have designed a silent self-stabilizing algorithm for constructing spanning trees of degree within one from the optimal. The algorithm is converging in polynomial time and uses registers on $O(\log n)$ bits. It stabilizes on FR-trees, which is a subclass of spanning trees with degree $\leq \opt+1$. It would be interesting to determine whether FR-trees is the ultimate class of trees with degree $\leq \opt+1$ for which there is  an efficient self-stabilizing algorithm. We proved that, unless $\NP=\coNP$, we cannot expect designing a polynomial-time silent self-stabilizing algorithm that stabilizes on the class of all spanning trees with degree $\leq \opt+1$. Nevertheless, there might be a subclass of trees with degree $\leq \opt+1$ including FR-trees, for which an efficient silent self-stabilizing algorithm exists. Note that it may be desirable that this subclass was including all spanning trees of optimal degree. We let this question as an open problem. 

Another direction for further work is to figure out whether or not it is possible to design a silent self-stabilizing algorithms constructing and stabilizing on (some subclass of) spanning trees with degree at most $\opt+1$, performing in $O(n)$ rounds. This seems to be a non-trivial issue. In fact, it is not even clear whether it is possible to construct a minimum-degree spanning tree with degree at most $\opt+1$ in time $o(n^2)$ in \emph{synchronous} models such as the $\cal CONGEST$ model of~\cite{P00}. The  algorithm by F\"urer and Raghavachari~\cite{FR94} appears to be  inherently sequential, in the sense the degree of the current tree decreases sequentially. Moreover, it does not appear clear how to use efficiently the ability to perform  updates in parallel. Thus, the design of an efficient distributed algorithm may require to come up with a new algorithm, and/or to relax the performances (e.g., targeting a degree $O(\opt)$ or $\opt+O(\log n)$). 

\paragraph{Acknowledgement:} The second author is thankful to David Peleg for  informative discussions  about minimum-degree spanning tree construction in the $\cal CONGEST$ model.

%%%%%%%%%%%%%%%%%%%%%%%%%%%%%%%%%%%%%%%%%%%%%
%\newpage
%\setcounter{page}{1}
%\pagenumbering{roman}
%%%%%%%%%%%%%%%%%%%%%%%%%%%%%%%%%%%%%%%%%%%%%
\bibliographystyle{plain}

%%%%%%%%%%%%%%%%%%%%%%%%%%%%%%%%%%%%
\end{document}